\documentclass[12pt]{article}

\usepackage[margin=1.00in]{geometry}
\usepackage[T1]{fontenc}
\usepackage{newtxtext}
\usepackage{microtype}
\usepackage{amsmath,amssymb,amsthm,mathtools,bm}
\usepackage{booktabs}
\usepackage{enumitem}
\usepackage{setspace}
\usepackage{titlesec}
\usepackage{xcolor}
\usepackage[round,authoryear]{natbib}
\usepackage{graphicx}
\usepackage{caption}
\usepackage{tikz}
\usetikzlibrary{patterns}
\usepackage[hidelinks]{hyperref}

\AtBeginDocument{\setlength{\abovedisplayskip}{5pt plus 2pt minus 2pt}%
\setlength{\belowdisplayskip}{5pt plus 2pt minus 2pt}%
\setlength{\abovedisplayshortskip}{2pt plus 1pt}%
\setlength{\belowdisplayshortskip}{2pt plus 1pt}}
\titleformat{\section}{\large\centering\scshape}{\thesection.}{0.6em}{}
\titleformat{\subsection}{\normalsize\itshape}{\thesubsection.}{0.6em}{}
\titlespacing*{\section}{0pt}{1.2em}{0.8em}
\titlespacing*{\subsection}{0pt}{0.9em}{0.4em}
\renewenvironment{abstract}{\begin{center}\begin{minipage}{0.86\textwidth}\small\setstretch{1.25}}{\end{minipage}\end{center}}

\newtheoremstyle{ecmathm}{7pt}{7pt}{\itshape}{}{\scshape}{:}{.5em}{\thmname{#1}\thmnumber{ #2}\thmnote{ (#3)}}
\theoremstyle{ecmathm}
\newtheorem{theorem}{Theorem}
\newtheorem{proposition}{Proposition}
\newtheorem{lemma}{Lemma}
\newtheorem{corollary}{Corollary}
\newtheoremstyle{conditionstyle}{7pt}{7pt}{\normalfont}{}{\scshape}{:}{.5em}{\thmname{#1}\thmnumber{ #2}\thmnote{ (#3)}}
\theoremstyle{conditionstyle}
\newtheorem{assumption}{Assumption}
\newtheorem{condition}{Condition}
\newtheorem{remark}{Remark}
\newtheorem{example}{Example}

\newcommand{\E}{\mathbb{E}}
\newcommand{\Prob}{\mathbb{P}}

\newcommand{\ind}{\mathbf{1}}
\newcommand{\high}{H}
\newcommand{\low}{L}

\title{\textbf{The Privacy Externality of Disclosing Correlated Data}}
\author{Rui Sun\thanks{Haas School of Business, University of California, Berkeley; ruisun233@berkeley.edu. I thank Sunil Dutta, Omri Even-Tov, Shawn Kim, Yaniv Konchitchki, Panos Patatoukas, Xiao-Jun Zhang, and seminar participants at Berkeley Haas for helpful comments. All errors are my own.}}
\date{June 2026}

\begin{document}
\maketitle

\begin{abstract}
\noindent A firm that discloses data about one customer moves a downstream seller's belief about every correlated customer, pricing third parties it never transacts with. This privacy externality equals the change in downstream deadweight loss, is signed by which side of the pricing threshold a customer is on, and falls hardest on those just carried across. Disclosure is privately optimal on an open set of imperfect correlations; incentive compatibility prices it through a distorted allocation and rations the discount at the top under a continuum of types. Selling tips disclosure past a liquidity threshold; consent dominates both data minimization and laissez-faire.

\medskip
\noindent\textsc{Keywords}: information design, Bayesian persuasion, data markets, privacy externality, mechanism design. \\
\noindent\textsc{JEL}: D82, D83, L86, K20.
\end{abstract}

\section{Introduction}\label{sec:intro}

Firms learn about their customers by trading with them. A retailer learns what a shopper will pay an insurer; an architect who renovates a house learns what its owner will pay an interior designer; a bank learns what a borrower will pay a landlord. This information is valuable to the second firm, and the first can profit by sharing it---selling it outright, or disclosing it to win the customer a discount she will pay for. Sharing is not free. The customer anticipates the use of her record and shifts her purchases toward cheaper goods, and, because customers' data are correlated, the disclosure of one record moves what the market believes---and, for every similar customer it carries across a pricing threshold, what she pays.

A bilateral view of this trade sees only the first cost. The second is an externality: the firm does not bear the price change it imposes on a customer's correlated peers, and so does not internalize it. Disclosure is also more than a transfer of information. It is a screening instrument the firm wields alongside the price of its own good, with a benefit and a cost it trades off. Neither force appears in a model with a single, isolated customer.

In this paper, we study a firm's incentive to disclose data about its customers. When is disclosure privately optimal? And what does it do to the customers whose data are correlated with the disclosed record? To address these questions, we put forward a model in which a buyer trades sequentially with two sellers, each making a take-it-or-leave-it offer. The model consists of three elements: the buyer's valuations for the two goods, imperfectly correlated and resolved over time, so that he learns his second valuation only at the second trade; the first seller's commitment to a disclosure rule, a Bayesian persuasion of the second seller about the buyer's report; and a population of buyers whose types share a common component.

Our first result characterizes exactly when disclosure is privately optimal. Two forces close the region. Under perfect positive correlation the benefit vanishes: a buyer with a low taste today has surely a low taste tomorrow and assigns no value to a downstream discount. When instead the buyer's current taste is weak but the good itself is cheap to supply, a second incentive constraint binds: the discount makes the high-priced contract attractive to the low type, and the seller must distort the low type's own supply downward to deter him. The characterization prices this deterrence and shows it never closes the region on its own: disclosure remains optimal on the same open set of imperfectly correlated environments, but the mechanism that attains it changes, carrying a supply distortion alongside the information rent. The optimality of privacy obtained by the bilateral literature is the perfectly correlated boundary of this region, and the engine of the reversal is the sequential resolution of the buyer's tastes---with both valuations known at the first trade, privacy prevails under the same correlation structures (Section~\ref{sec:robust}).

The characterization is exact, so it contains its own converse: outside the region no informative signal raises revenue, perfect positive correlation included. This recovers the privacy result of \citet{CalzolariPavan2006} as a corollary. Under negative correlation the geometry reverses and a new margin appears: when the buyer's current taste is weak, the seller chooses between serving him---paying an information rent the disclosure itself reduces, by exactly the correlation gap---and excluding him to disclose rent-free; privacy is then strictly optimal on a positive-measure set of environments that the first best would assign to disclosure.

With a continuum of types the same pricing of incentive compatibility acquires a sharper form, and solving for it is the technical engine behind the externality result. Under a linear-exposure structure, incentive compatibility is equivalent to the monotonicity of a single index that bundles the trade allocation with the discount odds: the buyer cannot be given a discount schedule that falls in his report while his supply stands still, or he buys into the pool. The seller's problem reduces to a one-dimensional screening problem in that index, and its solution has a structure the binary model cannot display: the discount pool absorbs the \emph{bottom} of the type space, whose low correlation dilutes the downstream seller's posterior, while the types most eager for the discount are \emph{rationed}---they receive it with an interior, constant probability, paid for by full supply---because giving it to them outright would break the index's monotonicity. Disclosure is rationed at the top, not withheld, and deterministic reveal-or-pool rules are strictly suboptimal.

Our central result is the force a bilateral model cannot express. When customers' data are correlated, the disclosure that benefits the firm also moves the downstream seller's belief about every correlated customer---and, where the belief crosses the pricing threshold, her price and her surplus---an effect the firm does not internalize. The externality on each customer equals the disclosure-induced reduction in her expected downstream deadweight loss, and its incidence is asymmetric: customers the market already serves on favorable terms are harmed---most harmed when their exposure to the disclosed record just suffices to carry their price across the threshold, and progressively less as exposure deepens---while customers the market would otherwise distort gain, the more the more exposed they are. The privately optimal rule is socially optimal only when the exposed harmed mass is small; past an explicit threshold the planner prefers a truncated disclosure that preserves the firm's discount at zero population harm. Privacy itself is never the planner's choice.

The remaining results apply this account to the market for data and its regulation. Whether the firm may sell its data, rather than only act on it, is irrelevant for what is disclosed until the downstream seller's stake in pricing the population crosses an explicit liquidity threshold; past it, the optimum jumps to full revelation---the data market does not refine disclosure at the margin, it tips it. The corresponding remedy is not a choice between privacy and disclosure. A consent right held by each customer over the use of the record in pricing her implements the planner's buyer-by-buyer optimum: it harvests the externality's gains and truncates its harms, and so dominates both laissez-faire and a ban. Where only the signal's informativeness can be regulated, caps that dilute the signal are bang-bang, a cap on the favorable tail implements the planner's truncation, and in the continuum the optimal cap stops at the largest informativeness that keeps the exposed population inside its pricing region.

The externality is a phenomenon of downstream market power: under perfect competition prices do not respond to beliefs and it vanishes; under imperfect competition the same law governs its sign, while competition relocates the threshold and shrinks every customer's stake at the rate of the downstream margin.

The remainder of the paper proceeds as follows. After a leading example (Section~\ref{sec:example}) and the model (Section~\ref{sec:model}), the core analysis runs in three steps: optimal disclosure with a single buyer of binary types (Section~\ref{sec:single}), then of a continuum of types (Section~\ref{sec:engine}), and then the privacy externality on a correlated population (Section~\ref{sec:pop}). The remaining sections apply and extend this account: the market for data, consent, and data minimization (Section~\ref{sec:market}), and the robustness of the externality to downstream competition and to the timing of the buyer's information (Section~\ref{sec:robust}); Section~\ref{sec:conclusion} concludes. Proofs of the three theorems and of the single-buyer analysis are in the Appendix, the remaining proofs in the Supplemental Material, and every numerical value in the paper regenerates from a replication script.

\subsection{Related literature}\label{ss:relatedlit}

The paper is most closely connected to the economics of data externalities. \citet{AcemogluMMO2022} show that correlation across individuals leads platforms to acquire too much data, because each individual does not internalize the informativeness of her data about others; \citet{ChoiJeonKim2019} and \citet{BergemannBonattiGan2022} develop related externalities in data collection and social data; these correlation-based acquisition externalities are uniformly signed. \citet{GalpertiLevkunPerego2024} value individual records by their externality on other records in an intermediary's pooling problem---a two-sided accounting of the intermediary's own value that arises even from statistically independent records. Ours operates on the \emph{disclosure} of data already held, through downstream pricing of customers who never transact with the disclosing firm: it is a welfare incidence on correlated third parties, two-sided, signed by the side of the pricing threshold on which the affected customer stands, summarized by one sufficient statistic---the reduction in downstream deadweight loss---and switched on by the data market itself (Proposition~\ref{prop:sell}). The sale of information as a product is studied by \citet{BergemannBonattiSmolin2018} and \citet{BergemannBonatti2024}; here the commodity the sale prices is inference about third parties. The market for customer information descends from \citet{Taylor2004}, whose seller trades raw purchase histories; here the firm commits to a persuasion rule and the histories of others are what the sale prices. \citet{BonattiCisternas2020} study scores that aggregate histories for downstream pricing; \citet{DovalSkreta2025} study a seller curbing the ratchet by coarsening its product line; \citet{IchihashiAER2020} studies consumer-side disclosure to a multiproduct seller. On the policy side, consent and data minimization are studied by \citet{ArgenzianoBonatti} in a data-linkage model where the consumer's own record follows her across firms; our consent result concerns third parties---each customer vetoes the use of the source's record in pricing her---and implements the planner's buyer-by-buyer optimum exactly because the wedge and the customer's private interest are aligned in sign. The welfare economics of privacy is surveyed by \citet{AcquistiTaylorWagman2016}; behavior-based pricing, the reduced-form ancestor of our downstream stage, by \citet{FudenbergVillasBoas2006}. \citet{StrackYang2024} characterize signals that preserve privacy as a constraint; here privacy is an equilibrium outcome, and its failure prices an externality. \citet{HidirVellodi2021} and \citet{AliLewisVasserman2023} study consumer-initiated disclosure to a price-discriminating seller; the disclosing party here is the firm, and the harmed parties are not in the room.

Disclosure here is a \emph{costly screening instrument}, which places the paper next to the recent theory of multidimensional screening. \citet{YangCMS} shows that when an agent's preferences over a productive and a costly instrument are positively correlated, the costly instrument is not used and one-dimensional screening is optimal; our privacy result is the information-design counterpart, with the costly instrument generated endogenously by a downstream market rather than taken as an exogenous action, and with a cross-agent externality single-agent screening has no room for. Unlike the costly signals of \citet{Spence1973}, which the agent bears, here the principal exercises the instrument and, through correlation, its cost falls on third parties. The multiproduct-pricing reading of costly screening \citep{HaghpanahHartline2021,YangBundling} and the broader theory of multidimensional screening \citep{RochetStole2003} are the closest mechanism-design antecedents.

The paper builds on the sequential-contracting model of \citet{CalzolariPavan2006}, who study an agent contracting with two principals, the first committing to a disclosure rule. They restrict attention to perfect positive and perfect negative correlation, so the buyer knows both valuations at the first trade, and under perfect positive correlation they obtain that privacy is optimal. Theorem~\ref{thm:binary} locates that conclusion exactly: perfect positive correlation is the case excluded from our disclosure region, where the low type's benefit from a downstream discount vanishes while the high type's incentive cost is maximal. The deeper engine of the reversal is the sequential resolution of the buyer's tastes, in the sell-the-option tradition of \citet{EsoSzentes2007} and \citet{CourtyLi2000}: Proposition~\ref{prop:static} shows that when the buyer knows both valuations at the first trade, privacy is optimal under the same correlation structures, so what overturns the classical conclusion is precisely the buyer's uncertainty about his future taste interacting with imperfect correlation.

The continuous-type analysis connects to the price-theoretic approach to persuasion. The Lagrangian we solve is in the family of \citet{DworczakMartini2019} and \citet{Kolotilin2018}; what is new is the screening side: the sender's posteriors are constrained by the incentive compatibility of her own mechanism, and the resulting reduction---a single monotone index bundling allocation and discount odds---and its rationed-boundary solution have no counterpart in pure persuasion. The virtual weighting of the disclosure objective is the persuasion counterpart of the dynamic virtual values of \citet{EsoSzentes2017} and \citet{PavanSegalToikka2014}, and the ironing technology traces to \citet{Myerson1981} and \citet{Toikka2011}. \citet{Dworczak2016} studies mechanism design with an aftermarket whose payoff depends on a public posterior about the winner; our second seller plays the role of his aftermarket, but our upstream seller actively persuades, and our central object is the cross-agent externality his single-agent analysis does not contain. \citet{BergemannHeumannMorris2026} jointly design a menu and the buyer's information; here the persuasion is of a third party about the buyer, and the friction is what that third party's response does to everyone else.

Finally, the disclosure rule is a Bayesian persuasion in the sense of \citet{KamenicaGentzkow2011}, part of the information-design program surveyed by \citet{BergemannMorris2019}, and we use their concavification to solve the binary problem; the sequential resolution of the buyer's preferences places the model in the dynamic-mechanism-design tradition of \citet{PavanSegalToikka2014} and \citet{Battaglini2005}.

\section{An Illustrative Example}\label{sec:example}

A concrete example fixes the forces before the general model. An online retailer ($S_1$) makes a take-it-or-leave-it offer for its premium product; an insurer ($S_2$) later quotes the same customer a policy. The customer values each purchase at $\high=3$ or $\low=1$, and the two values are correlated through his taste for quality: a customer who pays premium prices probably buys generous coverage. He learns his policy value only when the quote arrives, so at the retail stage he knows only his retail type, high with probability $\mu_0=1/2$. The two types are linked by $\bar\high=\Prob(\text{policy high}\mid\text{retail high})=1/2$ and $\bar\low=\Prob(\text{policy high}\mid\text{retail low})=1/4$.

Absent any data sharing, the insurer's belief that the customer values coverage highly is $\mu_0\bar\high+(1-\mu_0)\bar\low=3/8$, above the threshold $q^*=\low/\high=1/3$ at which it is indifferent between the high premium $\high$ (serving only high types) and the low premium $\low$ (serving both). It therefore quotes $\high=3$, and the customer sees no discount. The retailer, anticipating this, prices the product at $3$ and sells only to the high type, earning $\mu_0\cdot 3=3/2$; the low retail type, worth $\low=1$ but unwilling to pay $3$, is excluded and contributes nothing.

Now let the retailer share data with the insurer and commit in advance to the sharing protocol: it reveals a high retail type with probability $g_\high=1/2$ and otherwise files the customer in a pooled segment. Upon seeing the pooled segment the insurer's belief that the customer is a high policy type falls to exactly $q^*=1/3$, so it is willing to discount, quoting $\low=1$. Disclosure manufactures a downstream discount with positive probability, and this reshapes the retailer's revenue through two channels.

\emph{The benefit runs through the low type.} A low retail type is a high policy type with probability $\bar\low=1/4$ and then gains $\high-\low=2$ from the discount, so the prospect of a discount is worth $\bar\low(\high-\low)=1/2$ to him. The retailer, who previously extracted nothing from a low type, now sells him exactly this prospect---a paid membership whose only content is a place in the pooled segment---and captures the $1/2$. Weighted by the probability $1-\mu_0=1/2$ of facing a low type, the benefit is $1/4$.

\emph{The cost runs through the high type.} A high retail type values the discount more---he is a high policy type with probability $\bar\high=1/2>\bar\low$---so he is tempted to pose as a low type and secure the discount with probability one instead of the $1-g_\high=1/2$ he obtains when truthful. To keep him truthful the retailer leaves him the value of that prospect as an information rent rather than charging for it, so his price for the product is unchanged. The net gain therefore comes from the low types, whose prospect the retailer does sell, and revenue rises from $3/2$ to $7/4$. Disclosure is strictly optimal.

The same accounting shows when disclosure fails. Push the environment toward perfect positive correlation, $(\bar\high,\bar\low)=(1,0)$. The benefit vanishes---a low retail type is now surely a low policy type and assigns no value to the discount---while the rent is maximal, because a high retail type is surely a high policy type and values the discount most. Disclosure then strictly lowers revenue and the retailer prefers privacy: this is the privacy result of \citet{CalzolariPavan2006}, which Section~\ref{sec:single} locates as the perfectly correlated boundary of a region of imperfect correlations on which disclosure is instead optimal. In this example the retailer serves only high retail types, so the discount never tempts an excluded customer; when instead the product is cheap relative to the customer's tastes, a second constraint awakens---the discounted policy makes the high type's contract attractive to the low type---and the retailer must distort what it sells him to keep the two contracts apart. Three questions remain, which the paper answers in turn: exactly which correlation structures make disclosure optimal, and at what incentive cost; what disclosure about one customer does to the other customers whose data are correlated with his; and what changes when the retailer can sell the data rather than only act on it.
\section{Model}\label{sec:model}

There are two periods $t\in\{1,2\}$, two sellers $S_1$ and $S_2$, and a buyer. Each seller supplies one indivisible good at zero cost. In period~$t$ seller $S_t$ makes a take-it-or-leave-it offer. All parties are risk neutral and outside options are zero. The buyer's value for $S_t$'s good is $\theta_t\in\{\low,\high\}$ with $\high>\low>0$, and the public prior is $\mu_0=\Prob(\theta_1=\high)$.

Valuations are imperfectly correlated across the two goods, and the buyer does not know $\theta_2$ when he contracts with $S_1$:
\begin{equation}
  \bar\high=\Prob(\theta_2=\high\mid\theta_1=\high),\qquad
  \bar\low=\Prob(\theta_2=\high\mid\theta_1=\low),
\end{equation}
with $\bar\high\ge\bar\low$ as the leading case, so a high first-period type is more likely to remain high. The pair $(\bar\high,\bar\low)$ indexes the correlation structure; $\bar\high=1,\bar\low=0$ is perfect positive correlation and $\bar\high=\bar\low$ is independence.

\begin{assumption}[Environment]\label{ass:env}
$\high>\low>0$, $\mu_0\in(0,1)$, and $\bar\high,\bar\low\in[0,1]$.
\end{assumption}

\noindent The positive-correlation case $\bar\high\ge\bar\low$ organizes the exposition; the negative-correlation case $\bar\high<\bar\low$ is treated in full in Proposition~\ref{prop:negcorr}. The analysis covers every correlation structure, including the perfectly correlated boundary of \citet{CalzolariPavan2006}.

The interaction between $S_1$ and the buyer is private: neither the report nor the trade is observed by $S_2$. By committing ex ante to a signal $s\in\mathcal S$ with distribution $g(s\mid b_1)$ conditioned on the reported type $b_1$, $S_1$ can nonetheless convey information to $S_2$. The revelation principle applies, so in period~1 $S_1$ offers the mechanism $\{p_1(b_1),q_1(b_1),g(s\mid b_1)\}$ specifying a price, a trade probability, and the conditional signal distribution. After observing $s$, $S_2$ offers $\{p_2(b_2,s),q_2(b_2,s)\}$, where the dependence on $s$ enters through $S_2$'s posterior about $\theta_2$. Given a posterior $q=\Prob(\theta_2=\high)$ about the buyer she faces, $S_2$ optimally posts the high price $\high$, selling only to high types, when $q>q^*\equiv\low/\high$, and the low price $\low$, selling to both types---a discount---when $q\le q^*$. One convention governs every tie in the paper: at $q=q^*$ exactly, $S_2$ is indifferent and discounts, the sender-preferred selection standard in persuasion; the downstream deadweight loss is accordingly $D(q)=\low\,(1-q)\,\ind\{q>q^*\}$.

A contract is incentive compatible if truthful reporting is optimal for the buyer given his beliefs, and individually rational given the zero outside option. Seller $S_1$ maximizes expected revenue from its own trade subject to incentive compatibility, individual rationality, and the downstream best response of $S_2$. We say that \emph{disclosure} occurs when the optimal $g(\cdot\mid b_1)$ is informative about $b_1$, and \emph{privacy} obtains when $g$ is uninformative. Two mechanisms are \emph{equivalent} if they induce the same joint distribution over the allocation, the transfer, and $S_2$'s posterior, and a solution is \emph{essentially unique} if every optimal mechanism is equivalent to it---pinning the allocation, rents, and what the market learns, while leaving the labelling and duplication of messages free.

Two readings of $S_1$'s problem recur. As \emph{information design}, $S_1$ is a sender who persuades the receiver $S_2$ about the buyer's type subject to the buyer's incentive constraints; this is the reading we exploit technically, through the concavification of \citet{KamenicaGentzkow2011} in the binary model and a constrained linear program in the continuum. As \emph{multidimensional screening}, the disclosure rule is a second screening instrument layered on the price and quantity of $S_1$'s own good: a buyer who reports a type buys both that good and a lottery over downstream discounts, and the two valuations he holds for these two objects are the two dimensions being screened. The instrument is costly---it distorts the report, the allocation that deters the report, and, once buyers are correlated (Section~\ref{sec:pop}), the welfare of third parties---so the question of when $S_1$ discloses is the question of when a costly screening instrument earns its keep \citep{YangCMS}.

\subsection{Discussion of the assumptions}\label{ss:assumptions}

Five features of the model carry the analysis. Each is the weakest version that delivers the results, and each is relaxed or shown inessential below.

\noindent\textbf{Binary valuations.} The two-point type space is the leading case, not a substantive restriction, and it is relaxed twice. Section~\ref{sec:engine} solves the continuum-type problem in full, where the binary reveal-or-pool structure gives way to a rationed discount; Section~\ref{sec:pop} states the externality representation for an arbitrary finite downstream type space, the belief entering only through the deadweight-loss function.

\noindent\textbf{Conditional independence (Assumption~\ref{ass:pop}).} Buyers' first-period types are independent given a common state $\omega$, and the cross-buyer posterior loads on the source linearly. In the binary model the linear loading is implied by conditional independence, and any two-message signal satisfies it automatically (Remark~\ref{rem:colinear}); with richer signals it is the one-factor structure used throughout the data-externality literature \citep{AcemogluMMO2022,BergemannBonattiGan2022}, and Example~\ref{ex:loading} shows it is load-bearing.

\noindent\textbf{Monopoly pricing downstream.} Seller $S_2$ posts a take-it-or-leave-it price. The analysis uses only that her optimal price is monotone in her belief, so the buyer's continuation payoff is a monotone step. Downstream competition is not a hostage but a comparative static: Proposition~\ref{prop:competition} shows the externality's law---harm below the prevailing pricing threshold, gain above---holds for every posted-price market structure, while competition relocates the threshold and scales every stake by the downstream margin, to zero under Bertrand pricing.

\noindent\textbf{Commitment to the disclosure rule.} Seller $S_1$ commits ex ante to the signal $g$, as in the Bayesian-persuasion paradigm of \citet{KamenicaGentzkow2011}. Without commitment the disclosure stage is cheap talk and conveys nothing; commitment is what gives the instrument any bite.

\noindent\textbf{Sequential learning.} The buyer learns $\theta_2$ only when he meets $S_2$. This is essential, not cosmetic: Section~\ref{sec:robust} (Proposition~\ref{prop:static}) shows that if the buyer knows both valuations at the first trade, privacy prevails---proven at the perfectly correlated benchmark and whenever $\bar\low\ge q^*$, the residual region recorded computationally in Remark~\ref{rem:staticregion}. The value of disclosure comes precisely from the buyer's uncertainty about his future taste---the feature the static benchmark of \citet{CalzolariPavan2006} assumes away.

\section{Optimal Disclosure with a Single Buyer}\label{sec:single}

We characterize $S_1$'s optimal disclosure to a single buyer. After recording the buyer's continuation payoff, we solve the first best, in which $\theta_1$ is observable to $S_1$, and then the second best, in which $\theta_1$ is private and both incentive constraints are priced.

\subsection{The buyer's continuation payoff}

Let $\mu_1=\Prob(\theta_1=\high\mid s)$ be $S_2$'s posterior about the first-period type after the signal, and $q(\mu_1)=\bar\low+\mu_1(\bar\high-\bar\low)=\Prob(\theta_2=\high\mid s)$ the implied posterior about the second-period type. Seller $S_2$ discounts---posts $\low$ and sells to both types---if and only if $q(\mu_1)\le q^*$, equivalently $\mu_1\le\pi$, where
\begin{equation}\label{eq:pi}
  \pi\equiv\frac{q^*-\bar\low}{\bar\high-\bar\low},\qquad q^*\equiv\frac{\low}{\high}.
\end{equation}

\begin{lemma}[Continuation payoff]\label{lem:cont}
A low second-period type earns zero in the second trade. A high second-period type earns $\high-\low$ when $S_2$ discounts and zero otherwise, so the buyer's continuation payoff is $Q(\mu_1)=(\high-\low)\,\ind\{\mu_1\le\pi\}$. Conditional on a first-period report $b_1$, the buyer's expected continuation payoff is $\bar\low\,\E_g[Q\mid b_1]$ if $\theta_1=\low$ and $\bar\high\,\E_g[Q\mid b_1]$ if $\theta_1=\high$.
\end{lemma}

The high first-period type weights the downstream discount by $\bar\high$ and the low type by $\bar\low<\bar\high$: the buyer more likely to remain high values the prospect of a discount more. This single inequality drives every result below. Throughout, write $\phi_b=\E_g[\ind\{\mu_1\le\pi\}\mid b_1=b]$ for the discount probability a report induces; a signal is feasible exactly when its discount messages each carry posterior at most $\pi$, which aggregates to $\mu_0(1-\pi)\phi_\high\le\pi(1-\mu_0)\phi_\low$, that is, $\phi_\high\le k\,\phi_\low$ with $k\equiv\pi(1-\mu_0)/[\mu_0(1-\pi)]$, whenever the prior is not itself discounted ($\mu_0>\pi$).

\subsection{First best}

Suppose $\theta_1$ is observable to $S_1$, isolating the benefit of disclosure from any incentive cost.

With both individual-rationality constraints binding, substituting the transfers into $S_1$'s revenue leaves
\begin{equation}\label{eq:fbobj}
  \max_{q_1,\,g}\ \underbrace{\mu_0\high\,q_1(\high)+(1-\mu_0)\low\,q_1(\low)}_{\text{first-period trade}}+\underbrace{\mu_0\bar\high\,\E_g[Q\mid\high]+(1-\mu_0)\bar\low\,\E_g[Q\mid\low]}_{\text{value of disclosure}} .
\end{equation}
The trade term is maximized by supplying both types, $q_1\equiv1$. Expressing the disclosure term through the unconditional distribution of posteriors $\tilde g$ (Bayes-plausible, $\E_{\tilde g}[\mu_1]=\mu_0$) turns it into $\E_{\tilde g}[J_f(\mu_1)]$ with the pointwise value
\begin{equation}\label{eq:Jf}
  J_f(\mu_1)=Q(\mu_1)\,q(\mu_1)=(\high-\low)\,q(\mu_1)\,\ind\{\mu_1\le\pi\} .
\end{equation}
Concavifying $J_f$ delivers the first best.

\begin{proposition}[First-best disclosure]\label{prop:fb}
Under Assumption~\ref{ass:env} with $\bar\high\ge\bar\low$, $\bar\low\neq q^*$, and $\theta_1$ observable, informative disclosure is strictly optimal if and only if
\begin{equation}\label{eq:fbregion}
  \bar\low<\frac{\low}{\high}<\mu_0\bar\high+(1-\mu_0)\bar\low .
\end{equation}
When \eqref{eq:fbregion} holds, the optimal signal randomizes between revealing the high type (posterior $\mu_1=1$) and a pool (posterior $\mu_1=\pi$) that just induces a downstream discount, which $S_1$ monetizes through its trade with the low type---whose trade was otherwise worthless.
\end{proposition}

The left inequality in \eqref{eq:fbregion} says a revealed low type is discounted ($q(0)=\bar\low<q^*$); the right says that absent disclosure $S_2$ posts the high price ($q(\mu_0)>q^*$). Disclosure is valuable precisely when it manufactures a discount that does not otherwise exist, which fails once correlation is so positive that $\bar\low\ge\low/\high$.

\subsection{Second best}\label{ss:secondbest}

When $\theta_1$ is private, persuasion of $S_2$ feeds back into $S_1$'s own screening, and it does so through \emph{both} incentive constraints. The familiar constraint is the high type's: every informative signal gives the low report a weakly more frequent discount, $\phi_\low\ge\phi_\high$, and since the high type weights the discount by the larger $\bar\high$, he must be left a rent to deter mimicry. The second constraint is the low type's, and it awakens exactly when $S_1$ wants to serve him: a high contract that bundles full supply with a priced-in discount prospect is cheap---the rent makes it so---and the low type is tempted to take it. Feasibility requires the two contracts to stay apart,
\begin{equation}\label{eq:M}
  q_1(\high)+(\bar\high-\bar\low)\,\phi_\high\;\ge\;q_1(\low)+(\bar\high-\bar\low)\,\phi_\low,
  \tag{M}
\end{equation}
a monotonicity constraint bundling supply with discount odds that reappears, as the monotonicity of a single index, in the continuum analysis of Section~\ref{sec:engine}. Pricing both constraints yields the characterization.

\begin{theorem}[Optimal disclosure with binary types]\label{thm:binary}
Under Assumption~\ref{ass:env} with $\bar\high>\bar\low$, $\bar\low\neq q^*$ (the excluded slice is settled in Remark~\ref{rem:boundary}), and $\theta_1$ private, informative disclosure is strictly optimal if and only if
\begin{equation}\label{eq:sbregion}
  \max\{\bar\low,\ \mu_0\bar\high\}<\frac{\low}{\high}<\mu_0\bar\high+(1-\mu_0)\bar\low .
\end{equation}
On this region the optimal signal is essentially unique and coincides with the first-best two-signal rule: reveal the high type with probability $g_\high$ as in \eqref{eq:optrule}, else pool at the posterior $\pi$. The high type is supplied fully, $q_1(\high)=1$, with a strictly positive rent, and the low type's supply turns on the cost of keeping the contracts apart:
\begin{itemize}[leftmargin=1.6em,itemsep=0.1em,topsep=0.2em]
\item[(i)] if $\mu_0\ge\low/\high$, then $q_1(\low)=0$, \eqref{eq:M} is slack, only the high type's constraint binds, and the rent is $(\high-\low)(\bar\high-\bar\low)$; at $\mu_0=\low/\high$ a continuum of optimal supplies coexists, the signal still essentially unique;
\item[(ii)] if $\mu_0<\low/\high$, then $q_1(\low)=1-(\bar\high-\bar\low)\,g_\high\in(0,1)$, both incentive constraints bind, the rent is $(\high-\low)\,[\,q_1(\low)+\bar\high-\bar\low\,]$, and disclosure nets the deterrence cost $(\low-\mu_0\high)(\bar\high-\bar\low)\,g_\high$; the identity
\begin{equation}\label{eq:identity}
\mu_0(\high-\low)\big(q^*-\mu_0\bar\high\big)-(\low-\mu_0\high)\big(q(\mu_0)-q^*\big)
=\high(1-\mu_0)\big[\mu_0q^*(1-\bar\high)+(q^*-\mu_0)(q^*-\bar\low)\big]
\end{equation}
holds with strictly positive right side throughout \eqref{eq:sbregion}, so the deterrence cost never reverses the comparison.
\end{itemize}
Outside \eqref{eq:sbregion}, no informative signal raises revenue: when $q(\mu_0)>q^*$, privacy strictly dominates every discount-inducing signal except on the boundary $\mu_0\bar\high=q^*$, where they tie; when $q(\mu_0)\le q^*$ the prior itself is discounted.
\end{theorem}

Three readings. First, the region is the first-best region intersected with $\mu_0\bar\high<q^*$: the rent shrinks the region, while the deterrence---identity \eqref{eq:identity}---prices the mechanism without ever closing it further. Privacy under imperfect correlation is driven entirely by the classical rent. Second, the two cases of the theorem are two incentive costs of the same instrument: where the low type is excluded anyway, disclosure costs only the rent; where he is served, it also costs a supply distortion, $q_1(\low)$ falling by exactly the mimicry gain $(\bar\high-\bar\low)g_\high$ the discount creates. Third, because \eqref{eq:sbregion} is an exact characterization, the condition is necessary as well as sufficient: outside the region no informative signal raises revenue, perfect positive correlation included. This recovers the benchmark of the literature as a corollary.

\begin{corollary}[Calzolari and Pavan, 2006]\label{cor:cp}
At perfect positive correlation $(\bar\high,\bar\low)=(1,0)$, privacy is optimal: no informative signal raises $S_1$'s revenue. This is the privacy result of \citet{CalzolariPavan2006}, recovered as the perfectly correlated case excluded by \eqref{eq:sbregion}.
\end{corollary}

\begin{remark}[The Calzolari--Pavan privacy benchmark]\label{rem:cp}
At perfect positive correlation $(\bar\high,\bar\low)=(1,0)$ the region \eqref{eq:sbregion} is empty: with $\bar\low=0$ the incentive bound $\mu_0\bar\high$ and the no-disclosure belief $\mu_0\bar\high+(1-\mu_0)\bar\low$ both equal $\mu_0$, leaving no room for $\low/\high$ between them. Economically, the benefit of disclosure, borne by the low type and proportional to $\bar\low$, vanishes, while the incentive cost, borne by the high type and proportional to $\bar\high$, is maximal. Privacy is then strictly optimal, recovering \citet{CalzolariPavan2006}: their perfect-correlation assumption removes exactly the buyer's uncertainty about $\theta_2$ that makes a downstream discount valuable. \eqref{eq:sbregion} places their conclusion within a characterization on which disclosure is optimal across an open region of imperfect correlations.
\end{remark}

\begin{remark}[The boundary $\bar\low=q^*$]\label{rem:boundary}
On the measure-zero slice $\bar\low=q^*$ exactly, with $\bar\high>\bar\low$, $\pi=0$ and the pool degenerates to full revelation of the low type, which the tie-breaking convention discounts. In the first best, disclosure is then strictly optimal whenever $q(\mu_0)>q^*$; in the second best, iff $\mu_0\bar\high<q^*$ and, when $\mu_0<q^*$, $\bar\high<1$. No result below depends on this slice.
\end{remark}

\medskip
The optimal rule has an explicit two-signal form, which we use below. When disclosure is optimal, $q(\mu_0)=\bar\low+\mu_0(\bar\high-\bar\low)>q^*$, and $S_1$ reveals a high first-period type with probability $g_\high$ and pools otherwise, with $g_\high$ chosen so that the pool posterior equals $\pi$ and just induces a discount:
\begin{equation}\label{eq:optrule}
  g_\high=1-\frac{\pi(1-\mu_0)}{\mu_0(1-\pi)},\qquad \pi \text{ as in \eqref{eq:pi}.}
\end{equation}
Under this rule $S_2$'s posterior about the buyer's first-period type is $\mu_1=1$ in the reveal event, which occurs with probability $\mu_0 g_\high$, and $\mu_1=\pi$ in the pool event; the pool posterior about $\theta_2$ equals $q^*$, so $S_2$ is indifferent and discounts.

\begin{example}\label{ex:single}
Let $\high=3$, $\low=1$, so $q^*=1/3$; let $\mu_0=1/2$ and $(\bar\high,\bar\low)=(1/2,1/4)$. Then $q(\mu_0)=3/8>q^*$, $\pi=1/3$, and $g_\high=1/2$: $S_1$ reveals a high type half the time and otherwise pools, lowering $S_2$'s posterior to the discount threshold; Figure~\ref{fig:region} draws the disclosure regions. Since $\mu_0=1/2\ge q^*$, this is case (i) of Theorem~\ref{thm:binary}: the low type is excluded and only the rent is paid. At $(\high,\low,\mu_0,\bar\high,\bar\low)=(3,1,3/10,9/10,1/10)$, inside \eqref{eq:sbregion} with $\mu_0<q^*$, case (ii) obtains: $g_\high=2/51$, the low type's supply is distorted to $q_1(\low)=1-(4/5)(2/51)=247/255$, and the deterrence cost $(\low-\mu_0\high)(\bar\high-\bar\low)g_\high=8/2550$ is under two percent of the gross disclosure gain.
\end{example}

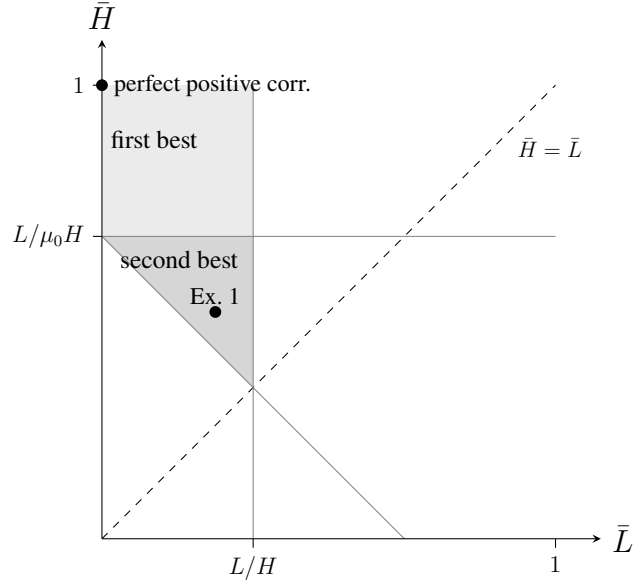
\begin{figure}[t]
\centering
\begin{tikzpicture}[scale=6,>=stealth]
  \def\qs{0.33333}
  \def\tq{0.66667}
  \fill[gray!16] (0,\tq) -- (0,1) -- (\qs,1) -- (\qs,\qs) -- cycle;
  \fill[gray!32] (0,\tq) -- (\qs,\tq) -- (\qs,\qs) -- cycle;
  \draw[dashed] (0,0) -- (1,1) node[pos=0.9,below right,scale=0.7] {$\bar\high=\bar\low$};
  \draw[very thin,gray] (\qs,0) -- (\qs,1);
  \draw[very thin,gray] (0,\tq) -- (1,\tq);
  \draw[very thin,gray] (0,\tq) -- (\tq,0);
  \draw[->] (0,0) -- (1.1,0) node[right] {$\bar\low$};
  \draw[->] (0,0) -- (0,1.1) node[above] {$\bar\high$};
  \draw (\qs,0) -- (\qs,-0.02) node[below,scale=0.75] {$\low/\high$};
  \draw (1,0) -- (1,-0.02) node[below,scale=0.75] {$1$};
  \draw (0,\tq) -- (-0.02,\tq) node[left,scale=0.75] {$\low/\mu_0\high$};
  \draw (0,1) -- (-0.02,1) node[left,scale=0.75] {$1$};
  \fill (0,1) circle (0.013) node[right,scale=0.75] {\,perfect positive corr.};
  \fill (0.25,0.5) circle (0.013);
  \node[scale=0.75] at (0.25,0.535) {Ex.~\ref{ex:single}};
  \node[scale=0.8] at (0.115,0.88) {first best};
  \node[scale=0.8] at (0.17,0.615) {second best};
\end{tikzpicture}
\caption{The disclosure region under positive correlation, drawn for $\mu_0=1/2$ and $q^*=\low/\high=1/3$. Disclosure is strictly optimal in the first best on the light region \eqref{eq:fbregion} and in the second best on the darker triangle \eqref{eq:sbregion}; the difference is the high type's rent. Perfect positive correlation $(\bar\low,\bar\high)=(0,1)$ lies in the first-best region but not the second best, so privacy is second-best optimal there (Remark~\ref{rem:cp}). The figure is drawn for $\mu_0\ge q^*$, case (i) of Theorem~\ref{thm:binary}; for $\mu_0<q^*$ the region is identical and the mechanism adds the supply distortion of case (ii). Example~\ref{ex:single} sits inside the second-best region.}
\label{fig:region}
\end{figure}

\subsection{Negative correlation and welfare}\label{ss:negcorr}

Negative correlation, $\bar\high<\bar\low$, reverses the geometry and redraws the incentive structure. Now $q(\mu_1)=\bar\low+\mu_1(\bar\high-\bar\low)$ is decreasing in $\mu_1$, so $S_2$ discounts when she believes the buyer is likely a high first-period type: $q(\mu_1)\le q^*$ iff $\mu_1\ge\pi$, with $\pi$ still given by \eqref{eq:pi}, and the feasibility of a signal aggregates to $\phi_\low\le\phi_\high/\tilde k$ with $\tilde k\equiv\pi(1-\mu_0)/[\mu_0(1-\pi)]>1$ whenever the prior is undiscounted ($\mu_0<\pi$). A high first-period type is now a high second-period type only with the small probability $\bar\high$, so he values the discount less than the low type; the mimicry that matters runs in the same direction as the supply, and which constraint binds turns on whether the low type is served.

\begin{proposition}[Disclosure under negative correlation]\label{prop:negcorr}
Suppose $\bar\high<\bar\low$ and $\bar\high<q^*<q(\mu_0)$. Whenever disclosure is optimal the optimal signal reveals the low first-period type with probability $g_\low=1-\mu_0(1-\pi)/[\pi(1-\mu_0)]$ and pools the rest at posterior $\pi$; $q_1(\high)=1$ throughout.
\begin{itemize}[leftmargin=1.6em,itemsep=0.1em,topsep=0.2em]
\item[(i)] If $\mu_0\ge\low/\high$, informative disclosure is strictly optimal everywhere on the regime, at zero information rent, with $q_1(\low)=(\bar\low-\bar\high)/\tilde k>0$ and both types' payoffs zero; at $\mu_0=\low/\high$ the rent-paying mechanism of (ii) ties.
\item[(ii)] If $\mu_0<\low/\high$, informative disclosure is strictly optimal if and only if
\begin{equation}\label{eq:negB}
(\high-\low)\Big[\mu_0\bar\high+\frac{\bar\low-\mu_0\bar\high}{\tilde k}\Big]\;>\;(\low-\mu_0\high)(\bar\low-\bar\high)\Big(1-\frac{1}{\tilde k}\Big),
\end{equation}
and privacy is strictly optimal on the open, positive-measure subset of the regime where \eqref{eq:negB} fails strictly. When it holds the optimal mechanism sets $q_1(\low)=1-(\bar\low-\bar\high)(1-1/\tilde k)$ and leaves the high type the rent $(\high-\low)[1-(\bar\low-\bar\high)]$, strictly dominating exclusion of the low type.
\end{itemize}
\end{proposition}

On this regime a discount is possible but not free: $q(\mu_0)>q^*$ means the prior is undiscounted, while $\bar\high<q^*$ means a sufficiently high-type pool can be pushed below the threshold. The pool posterior $\pi$ is exactly the belief at which $S_2$ is indifferent, so the pooled report just induces a discount. The contrast with Theorem~\ref{thm:binary} is instructive. In case (i) disclosure is entirely free of incentive cost---the discount flows to the report whose supply is already full---and the only trace of screening is the partial service of the low type, calibrated so the cheap low contract does not tempt the high type. In case (ii), serving the low type forces a rent the disclosure itself cuts by exactly $(\high-\low)(\bar\low-\bar\high)$; whether disclosure pays turns on \eqref{eq:negB}, and privacy survives on a positive-measure set of imperfectly negatively correlated environments---a margin the first best, which discloses throughout the regime, does not see.

\begin{remark}[Perfect negative correlation]\label{rem:perfneg}
At $(\bar\high,\bar\low)=(0,1)$ the regime reads $q^*<1-\mu_0$; outside it the prior is discounted and no informative signal raises revenue. Within it, in case (i), $\mu_0\ge q^*$, disclosure is strictly optimal---the disclosure case of \citet{CalzolariPavan2006} under perfect negative correlation; in case (ii) the condition \eqref{eq:negB} decides, and privacy can be optimal even here.
\end{remark}

\begin{corollary}[Welfare incidence]\label{cor:welfare}
Whenever disclosure is optimal, a low first-period type's participation constraint binds and he is exactly indifferent to disclosure. The high type's stake is signed by the correlation: under positive correlation he strictly gains---his rent rises from $0$ to $(\high-\low)(\bar\high-\bar\low)$ in case (i) of Theorem~\ref{thm:binary}, and from $(\high-\low)$ to $(\high-\low)[1+(\bar\high-\bar\low)(1-g_\high)]$ in case (ii). Under negative correlation he is indifferent in case (i) of Proposition~\ref{prop:negcorr}, and strictly loses in case (ii), where disclosure cuts his rent by exactly $(\high-\low)(\bar\low-\bar\high)$.
\end{corollary}

The welfare incidence falls entirely on the high type, and its sign---gain under positive correlation, loss under negative---turns on the same correlation structure that governs whether disclosure occurs. This within-buyer incidence is the seed of the cross-buyer externality of Section~\ref{sec:pop}: when buyers' types are correlated, the disclosure that helps or hurts $S_1$'s own counterparty also moves the price faced by every other buyer.
\section{Continuous Types: Rationed Disclosure}\label{sec:engine}

The binary model prices incentive compatibility through a rent and, when the low type is served, a supply distortion. With a continuum of types the same force acquires a sharper form, and solving for it is the technical heart of the paper. Let the buyer's first-period type be $\theta\in[0,1]$ with c.d.f.\ $F$, density $f>0$, equal to his valuation for $S_1$'s good, and let $\beta(\theta)=\Prob(\theta_2=\high\mid\theta)$ be the chance he values $S_2$'s good highly. After a signal $s$, $S_2$ holds a posterior over $\theta$ and discounts---posts $\low$---if and only if $\E[\beta(\theta)\mid s]\le q^*$.

\begin{condition}[Linear exposure]\label{cond:le}
$\beta(\theta)=\beta_0+\beta_1\theta$ with $\beta_1>0$, $\beta_0\ge0$, $\beta_0+\beta_1\le1$.
\end{condition}

\noindent Linear exposure is the continuum counterpart of the binary correlation structure---the two-point case is $\beta(0)=\bar\low$, $\beta(1)=\bar\high$---and the same affine form is what the population block of Section~\ref{sec:pop} assumes of cross-buyer inference, so one structure runs the whole paper. It is load-bearing, not cosmetic: Example~\ref{ex:nonaffine} exhibits a strictly convex $\beta$ under which the reduction below fails in both directions. We also maintain a strictly increasing virtual value, $\psi(\theta)=\theta-(1-F(\theta))/f(\theta)$ (implied by a strictly monotone hazard rate), and the disclosure-relevant regime $\beta(0)<q^*<\E[\beta(\theta)]$: a discount can be manufactured but does not occur freely.

Seller $S_1$ offers a direct mechanism $\{x(b),t(b),g(\cdot\mid b)\}$: a trade probability $x$, a payment $t$, and a signal $g$ to $S_2$, all conditioned on the reported type $b$. Writing $\phi(b)\in[0,1]$ for the discount probability the signal induces, the buyer of true type $\theta$ who reports $b$ obtains
\begin{equation}\label{eq:contpayoff}
  x(b)\,\theta - t(b) + (\high-\low)\,\beta(\theta)\,\phi(b).
\end{equation}

\subsection{A one-dimensional reduction}

Under linear exposure, \eqref{eq:contpayoff} reorganizes: with $c\equiv(\high-\low)\beta_1$,
\[
x(b)\,\theta+(\high-\low)\beta(\theta)\phi(b)-t(b)
=\theta\,\big[\underbrace{x(b)+c\,\phi(b)}_{\textstyle z(b)}\big]-\big[t(b)-(\high-\low)\beta_0\phi(b)\big].
\]
The report matters to the buyer only through the bundle $z=x+c\phi$ and an adjusted transfer. The two-instrument screening problem is therefore a one-dimensional screening problem in $z$, and the textbook characterization applies exactly.

\begin{lemma}[The $z$-reduction]\label{lem:zred}
Under Condition~\ref{cond:le}, a mechanism is incentive compatible if and only if $z=x+c\,\phi$ is nondecreasing and transfers satisfy the envelope formula. The buyer cannot be offered a discount schedule that falls in his report unless his supply rises by $c$ times the drop: a discount withdrawn at stationary supply invites every higher type to buy into the pool.
\end{lemma}

The lemma replaces, and explains, the failure of reveal-or-pool disclosure in the continuum: a rule that pools an interval and reveals the types above it drops $\phi$ from one to zero where $x$ has nowhere to rise, $z$ falls, and the types just above the pool deviate into it. What survives is governed by the reduced program. By the envelope formula and integration by parts, revenue at the participation-binding transfers equals $\E[x\,\psi]+(\high-\low)\,\E[\phi\,v]$, where
\begin{equation}\label{eq:revrep}
  \psi(\theta)=\theta-\frac{1-F(\theta)}{f(\theta)},\qquad
  v(\theta)=\beta(\theta)-\frac{1-F(\theta)}{f(\theta)}\,\beta'(\theta)
\end{equation}
are the virtual value of the trade and the \emph{virtual correlation} of the disclosure---the second-period counterpart, in the sense of \citet{EsoSzentes2017} and \citet{PavanSegalToikka2014}, of the Myersonian first term, here entering not an allocation but the persuasion of a third party. Feasibility of the signal aggregates exactly as in the binary model: every discount message must carry posterior $\beta$-mean at most $q^*$, which by pooling and splitting is equivalent to the single constraint $\E[(\beta(\theta)-q^*)\,\phi(\theta)]\le0$ on the regime $\E[\beta]>q^*$ (Lemma~\ref{lem:obed}). The second best therefore solves
\begin{equation}\label{eq:program}
  \max_{x,\phi\,\in[0,1]}\ \E\big[x\,\psi+(\high-\low)\,\phi\,v\big]
  \quad\text{s.t.}\quad z=x+c\,\phi\ \text{nondecreasing},\quad \E[(\beta-q^*)\,\phi]\le0 .
  \tag{P}
\end{equation}

Under linear exposure the two virtual objects obey an identity that does the heavy lifting: $(\high-\low)\,v-c\,\psi=(\high-\low)\,\beta_0$, a constant, for every type distribution. Attaching a multiplier $\lambda$ to obedience, the relative attractiveness of the discount over the trade at type $\theta$ is
\begin{equation}\label{eq:slam}
  s_\lambda(\theta)=(\high-\low)\,\beta_0-\lambda\,\big(\beta(\theta)-q^*\big),
\end{equation}
strictly decreasing in $\theta$ for every $\lambda>0$: the conflict between what the seller wants to disclose and what obedience permits is single-crossing by construction, with no auxiliary regularity. Types below the crossing $\theta_\lambda=\beta^{-1}(q^*+(\high-\low)\beta_0/\lambda)$---with the convention $\theta_\lambda=1$ when the inverse falls outside $[0,1]$---are cheap to include in the pool, those with $\beta$ below $q^*$ actively diluting it; types above the crossing are expensive.

\subsection{The structure of optimal disclosure}

\begin{theorem}[Rationed disclosure]\label{thm:engine}
Under Condition~\ref{cond:le}, a strictly increasing $\psi$, and $\beta(0)<q^*<\E[\beta]$, program \eqref{eq:program} attains the second best, characterized as follows.
\begin{itemize}[leftmargin=1.6em,itemsep=0.1em,topsep=0.2em]
\item[(i)] \emph{(Construction.)} Any $\lambda^*\ge0$ and feasible pair satisfying complementary slackness that maximize the Lagrangian over monotone bundles solve \eqref{eq:program}; for each $\lambda$, the Lagrangian maximizer is the monotone ironing of the pointwise problem, in which $z^*(\theta)$ maximizes $z\,\psi(\theta)+\phi\,s_{\lambda}(\theta)$ over $z\in[0,1+c]$ with $\phi$ at its admissible bound given $z$, and $z$ pooled at a constant where $z^*$ decreases.
\item[(ii)] \emph{(Anatomy.)} Let $F$ be uniform, $c\le1$, the margins conflict ($\theta^*<\theta_{\lambda^*}$, where $\psi(\theta^*)=0$), and the certainty pool be self-obedient, $\int_0^{\theta_{\lambda^*}}(\beta-q^*)\,dF\le0$. Then, essentially uniquely,
\[
\phi(\theta)=\begin{cases}1 & \theta\le\theta_{\lambda^*}\\ \bar\phi & \theta>\theta_{\lambda^*}\end{cases}
\qquad
x(\theta)=\begin{cases}0 & \theta<\theta^*\\ 1-c\,(1-\bar\phi) & \theta^*\le\theta\le\theta_{\lambda^*}\\ 1 & \theta>\theta_{\lambda^*},\end{cases}
\]
with $\bar\phi\in[0,1)$ and $z$ on a single plateau $\bar z=1+c\,\bar\phi$ above $\theta^*$. Here $(\lambda^*,\bar\phi)$ solves the plateau condition $\int_{\theta^*}^{1}\psi\,dF+\tfrac1c\int_{\theta_{\lambda^*}}^{1}s_{\lambda^*}\,dF=0$---under uniform $F$, $(1-\theta_{\lambda^*})^2=(\high-\low)/(2\lambda^*)$---and binding obedience $\int_0^{\theta_{\lambda^*}}(\beta-q^*)\,dF+\bar\phi\int_{\theta_{\lambda^*}}^{1}(\beta-q^*)\,dF=0$, with self-obedience equivalent to $\bar\phi\ge0$. If the plateau root violates self-obedience but the binding-obedience root is at least $\theta^*$ (uniform $F$: $2(q^*-\beta_0)/\beta_1\ge\theta^*$), the same anatomy holds with $\bar\phi=0$, $\theta_{\lambda^*}$ pinned by $\int_0^{\theta_{\lambda^*}}(\beta-q^*)\,dF=0$ and the plateau condition relaxed to an inequality. The remaining degenerate sub-case, in which that root too falls below $\theta^*$, is recorded in Remark~\ref{rem:bottompool}.
\item[(iii)] \emph{(Binary model.)} On a two-point support, monotonicity of $z$ is constraint \eqref{eq:M} and the plateau collapses to the supply distortion $q_1(\low)=1-(\bar\high-\bar\low)g_\high$ of Theorem~\ref{thm:binary}(ii).
\end{itemize}
\end{theorem}

\begin{remark}[Degenerate bottom pool]\label{rem:bottompool}
The worked anatomy of part (ii) covers the conflict case whenever the binding-obedience root $2(q^*-\beta_0)/\beta_1$ is at least $\theta^*$. When it falls below $\theta^*$ the certainty pool would reach past the Myerson cutoff; the same Lagrangian machinery applies and irons a single partial pool at the bottom of the type space, with binding obedience alone pinning its endpoint, while the rationing-at-the-top reading and the comparative statics are unchanged. The parametric environments studied below stay in the worked case, and we do not develop this degenerate one.
\end{remark}

The anatomy inverts the intuition that the seller discloses the types she most wants discounted. The discount pool absorbs the \emph{bottom} of the type space---including types whose virtual correlation is negative---because their low $\beta$ dilutes the pool's posterior and obedience, not virtue, is the scarce resource; when even the full bottom pool would breach obedience, the pool simply ends where its mean reaches the threshold and the rationing degenerates. The types above $\theta_{\lambda^*}$, who value the discount most and whom the relaxed problem would serve first, are \emph{rationed}: they receive the discount with the interior, constant probability $\bar\phi$, and they pay for it with full supply, because the only incentive-compatible way to give a higher type better discount odds is to bundle them into a $z$ the type below would not covet. On the pooled stretch the discount odds are read off the binding bound, $\phi$ sitting at its admissible level given the plateau bundle $z$. Disclosure is rationed at the top, not withheld; the seller's instrument is the odds, not the set.

\begin{proposition}[Deterministic disclosure is strictly suboptimal]\label{prop:converse}
Under the conditions of Theorem~\ref{thm:engine} with $\bar\phi\in(0,1)$, every deterministic disclosure rule---$\phi$ valued in $\{0,1\}$, in particular every rule that pools an interval of types and reveals the rest---is strictly suboptimal.
\end{proposition}

\begin{example}[The anatomy, exactly]\label{ex:engine}
Let $\high=5/2$, $\low=1$, so $q^*=2/5$; let $F$ be uniform and $\beta(\theta)=1/5+\theta/2$, so $c=3/4$, $\E[\beta]=9/20>q^*$, and $\beta(0)=1/5<q^*$. Solving the two pinning equations: $\lambda^*=3.822$, $\theta_{\lambda^*}=0.557$, $\bar\phi=0.404$, $\bar z=1.303$. The bottom $56$ percent of types are pooled into the discount with certainty; every type above is discounted with probability $0.404$ and supplied fully; types in $[0.5,0.557]$ are supplied at the interior level $0.553$, the trade the plateau leaves room for. The seller's objective attains $0.3589$; Figure~\ref{fig:engine} draws the anatomy. The best rule whose discount message pools a single interval---all remaining types pooled into the no-discount message---pools $[0.10,0.70]$ and attains $0.3100$, $13.6$ percent short. Deterministic rules of any form must split the pool---the best, with discount set $[0,0.565]\cup[0.877,1]$, attains $0.3501$---and remain $2.5$ percent short: the residual gap is the value of rationing. (A direct linear program over all incentive-compatible mechanisms on fine type grids reproduces the reduced program's value to machine precision and these magnitudes to the grid's resolution; see the replication script.)
\end{example}

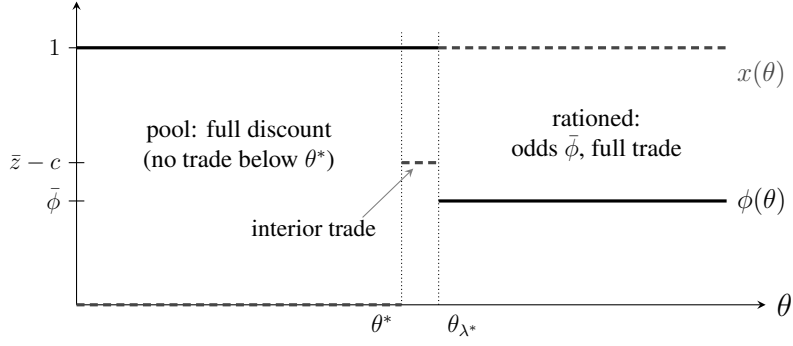
\begin{figure}[t]
\centering
\begin{tikzpicture}[xscale=8.6,yscale=3.4,>=stealth]
  \def\ts{0.5}\def\tl{0.557}\def\pb{0.404}\def\xb{0.553}
  \draw[->] (0,0) -- (1.06,0) node[right] {$\theta$};
  \draw[->] (0,0) -- (0,1.18);
  \draw[densely dotted] (\ts,0) -- (\ts,1.06);
  \draw[densely dotted] (\tl,0) -- (\tl,1.06);
  \node[below left,scale=0.78] at (\ts,-0.005) {$\theta^*$};
  \node[below right,scale=0.78] at (\tl,-0.005) {$\theta_{\lambda^*}$};
  \draw (-0.012,1) node[left,scale=0.75] {$1$} -- (0.012,1);
  \draw (-0.012,\pb) node[left,scale=0.75] {$\bar\phi$} -- (0.012,\pb);
  \draw (-0.012,\xb) node[left,scale=0.75] {$\bar z-c$} -- (0.012,\xb);
  \draw[very thick] (0,1) -- (\tl,1);
  \draw[very thick] (\tl,\pb) -- (1,\pb) node[right,scale=0.85] {$\phi(\theta)$};
  \draw[very thick,gray!55!black,densely dashed] (0,0) -- (\ts,0);
  \draw[very thick,gray!55!black,densely dashed] (\ts,\xb) -- (\tl,\xb);
  \draw[very thick,gray!55!black,densely dashed] (\tl,1) -- (1,1) node[below right,scale=0.85,yshift=-1pt] {$x(\theta)$};
  \node[scale=0.78,align=center] at (0.25,0.62) {pool: full discount\\(no trade below $\theta^*$)};
  \node[scale=0.75] at (0.365,0.30) {interior trade};
  \draw[->,thin,gray] (0.44,0.34) -- (0.518,0.53);
  \node[scale=0.78,align=center] at (0.80,0.66) {rationed:\\odds $\bar\phi$, full trade};
\end{tikzpicture}
\caption{Optimal disclosure with continuous types (Theorem~\ref{thm:engine}), at the parameters of Example~\ref{ex:engine}: uniform $F$, $\beta(\theta)=1/5+\theta/2$, $q^*=2/5$. The discount probability $\phi$ (solid) is one on the pool $[0,\theta_{\lambda^*}]$ and rationed at $\bar\phi=0.404$ above; the trade probability $x$ (dashed) is zero below the Myerson cutoff $\theta^*$, interior at $\bar z-c=0.553$ on the plateau's first stretch, and one above $\theta_{\lambda^*}$. The bundle $z=x+c\phi$ is constant at $\bar z=1.303$ from $\theta^*$ on: the plateau is what incentive compatibility costs. Values from the replication certificate.}
\label{fig:engine}
\end{figure}

\begin{example}[Linear exposure is load-bearing]\label{ex:nonaffine}
Let $\beta(\theta)=\theta^2$ on three types $\theta\in\{0,1/2,1\}$ and $\high-\low=1$. The allocation $x=(0.73,\,0.99,\,0.69)$, $\phi=(0.90,\,0.41,\,0.95)$ is incentive compatible---transfers exist by cyclical monotonicity---yet no constant $c$ makes $x+c\,\phi$ monotone, so the necessity direction fails; and the allocation $x=(0,\,0.05,\,1)$, $\phi=(1,\,1,\,0.05)$, whose bundle $x+\phi$ is monotone and which is implementable under the endpoint secant $\beta(\theta)=\theta$, violates cyclical monotonicity under $\beta=\theta^2$ (the middle type gains $0.2375$ in the two-cycle with the top type), so the sufficiency direction fails too. Under nonlinear exposure neither direction of Lemma~\ref{lem:zred} survives, and the one-dimensional reduction is genuinely a property of the linear-exposure class.
\end{example}

\begin{remark}[What replaces interval pooling]\label{rem:structure}
A natural conjecture, suggested by the relaxed problem, is that incentive compatibility carves the discount pool out of the interior of the type space, excluding low types whose virtual correlation is negative. The opposite is true at the full optimum: the bottom types are the pool's cheapest diluters and are included throughout the anatomy of Theorem~\ref{thm:engine}(ii); what incentive compatibility does is cap the \emph{odds} at the top. The conjecture fails because it prices the disclosure margin alone; the binding interaction is between the disclosure and the trade---the plateau in $z$.
\end{remark}
\section{Correlated Buyers and the Privacy Externality}\label{sec:pop}

The bilateral analysis treats the buyer as isolated. We now let $S_1$'s customer be drawn from a population whose data are correlated across individuals, and we trace the consequence of disclosure for those other individuals.

\subsection{A population of correlated buyers}

Buyers' types share a common component. Let $\omega$ be a common state and let there be a continuum of buyers $i\in[0,1]$ whose first-period types are conditionally independent given $\omega$, with $\Prob(\theta_1^i=\high\mid\omega)$ increasing in $\omega$. Seller $S_1$ trades with a single \emph{source} buyer, elicits her type through the mechanism of Section~\ref{sec:single}, and discloses a signal; because the source's type is informative about $\omega$, the signal shifts $S_2$'s belief about every buyer. Seller $S_2$ serves the whole population, pricing each buyer from her posterior.

\begin{assumption}[Population structure]\label{ass:pop}
Buyers' first-period types are conditionally independent given $\omega$; the signal is independent of $\omega$ conditional on the source's type; seller $S_2$ prices each buyer using only her own posterior; and each buyer $i$'s exposure is linear: with $\mu_A(s)=\Prob(\theta_1^A=\high\mid s)$ the posterior about the source,
\begin{equation}\label{eq:qi}
  q_i(s)=q_i^0+\rho_i\,(\bar\high-\bar\low)\,(\mu_A(s)-\mu_0),\qquad q_i^0\equiv\bar\low+\mu_i^0(\bar\high-\bar\low),
\end{equation}
where $q_i^0$ is $i$'s no-disclosure value and $\rho_i\ge 0$ her loading on the source. The population is described by the joint distribution $F$ of $(q_i^0,\rho_i)$, supported on pairs for which \eqref{eq:qi} is a posterior, $q_i(s)\in[0,1]$ for every $\mu_A\in[0,1]$; the source has $\rho=1$ and a buyer independent of the source has $\rho_i=0$. In the continuum model of Section~\ref{sec:engine} the same structure reads $q_i(s)=q_i^0+\rho_i\,(\E[\beta(\theta_A)\mid s]-\E[\beta(\theta_A)])$, the loading now on the source's posterior $\beta$-mean.
\end{assumption}

\begin{remark}[When linearity is free]\label{rem:colinear}
For any \emph{two-message} signal, Bayes plausibility places the two induced pairs (source posterior, buyer-$i$ posterior) and the prior pair on a common line, so the linear form \eqref{eq:qi} is automatic given conditional independence, with $\rho_i$ read off the two messages. Linearity is a substantive restriction only for richer signals, where it is the one-factor loading used throughout the data-externality literature \citep{AcemogluMMO2022,BergemannBonattiGan2022}.
\end{remark}

\begin{example}[Linearity is load-bearing beyond two messages]\label{ex:loading}
Let the source's type take three values with $\beta$-values $(0.2,0.4,0.6)$, each equally likely, and let a buyer's conditional means be $(0.1,0.5,0.6)$---conditionally independent given the state, but not affine in the source's $\beta$. In the environment of Example~\ref{ex:single}, under full revelation the buyer's true wedge is $0.300$, while the best linear-exposure approximation yields $0.283$: a planner pricing the externality through \eqref{eq:qi} misses a computable share of it. All results below hold exactly under Assumption~\ref{ass:pop} and approximately at its boundary.
\end{example}

\begin{proposition}[The private optimum is unchanged]\label{prop:private}
Seller $S_1$'s privately optimal disclosure rule is independent of the population $F$ and coincides with the single-buyer optimum. The positive theory of disclosure is the bilateral theory; all of the new content is normative.
\end{proposition}

\subsection{The wedge and its representation}

Fix a disclosure rule $g$. By \eqref{eq:qi} the signal induces, for each exposed buyer, a mean-preserving spread of $S_2$'s posterior $q_i$ about her. Let $\Sigma(q)$ denote buyer-side surplus---the buyer's consumer surplus plus $S_2$'s profit on her---when $S_2$'s posterior is $q$:
\begin{equation}\label{eq:sigma}
  \Sigma(q)=\underbrace{\low+(\high-\low)q}_{\text{full-trade surplus}}\;-\;\underbrace{\low(1-q)\,\ind\{q>q^*\}}_{\text{deadweight loss }D(q)},
\end{equation}
affine on each side of $q^*$ with a downward jump $D$ at the threshold: the loss from excluding the buyer's low type, who trades only under the discount. Buyer $i$'s wedge and the aggregate wedge are
\begin{equation}\label{eq:wedge}
  W_i(g)=\E_s\big[\Sigma(q_i(s))\big]-\Sigma(q_i^0),\qquad
  W(g;F)=\int W_i(g)\,dF(q_i^0,\rho_i).
\end{equation}
By Proposition~\ref{prop:private} the private value of $g$ is independent of $F$, so $W$ is the entire population-dependent component of the social value of disclosure. Its structure follows from a representation that does not depend on the binary parametrization at all.

\begin{proposition}[Wedge representation]\label{prop:dwl}
Let the downstream type space be finite, let the downstream seller choose a mechanism after forming a belief $P$ about the agent, and let $\Sigma(P)$ be buyer-side surplus under a fixed selection from her optimal mechanisms; write $\bar\Sigma(P)$ for the full-information surplus, linear in $P$, and $D\equiv\bar\Sigma-\Sigma\ge0$ for the downstream deadweight loss. Every Bayes-plausible disclosure inducing posteriors $\{P_i(s)\}$ with mean $P_i^0$ satisfies
\begin{equation}\label{eq:generalwedge}
  W_i=\E_s\big[\Sigma(P_i(s))\big]-\Sigma(P_i^0)=D(P_i^0)-\E_s\big[D(P_i(s))\big].
\end{equation}
If $D$ is concave (convex) on the convex hull of the induced posteriors, then $W_i\ge0$ ($\le0$); and no disclosure can harm agent $i$ if and only if $D$ coincides at $P_i^0$ with its concavification over the reachable posteriors. Under Assumption~\ref{ass:pop} the reachable posteriors form a line segment and the condition is one-dimensional; with binary downstream types it reduces to whether the segment crosses $q^*$: a buyer with $q_i^0>q^*$ is never harmed, and a buyer with $q_i^0\le q^*$ is harmed by some disclosure if and only if her exposure carries the segment strictly above $q^*$.
\end{proposition}

The wedge is thus the disclosure-induced reduction in the agent's expected downstream deadweight loss. The full-trade component of \eqref{eq:sigma} is linear in the belief, so disclosure moves buyer-side surplus only through $D$: it operates on the distortion, never on the pie. Everything else in this section is the tent-shaped $D$ of the binary model meeting the mean-preserving spreads of \eqref{eq:qi}.

\subsection{Asymmetric incidence}

Under the optimal rule $g^*$, the signal reveals the source ($\mu_A=1$) with probability $r=(\mu_0-\pi)/(1-\pi)$ and pools at $\mu_A=\pi$ otherwise. Buyer $i$'s posterior takes two values,
\[
q_i^R=q_i^0+\rho_i(\bar\high-\bar\low)(1-\mu_0),\qquad
q_i^P=q_i^0-\rho_i(\bar\high-\bar\low)(\mu_0-\pi),
\]
in the reveal and pool events respectively.

\begin{theorem}[Asymmetric incidence]\label{thm:incidence}
Under Assumptions~\ref{ass:env}--\ref{ass:pop}, on the disclosure region \eqref{eq:sbregion}, and under the rule $g^*$:
\begin{itemize}[leftmargin=1.6em,itemsep=0.1em,topsep=0.2em]
\item[(i)] \emph{(Closed forms.)} Every buyer's wedge is
\begin{align*}
W_i&=-\,r\,\low\,(1-q_i^R)\,\ind\{q_i^R>q^*\} && \text{if } q_i^0\le q^*,\\
W_i&=(1-r)\,\low\,(1-q_i^P)\,\ind\{q_i^P\le q^*\} && \text{if } q_i^0>q^*.
\end{align*}
The wedge is nonzero exactly when disclosure carries the market's belief across the pricing threshold and the crossing destroys or restores a trade: for buyers below $q^*$, exposure above $\underline\rho_i=(q^*-q_i^0)/[(\bar\high-\bar\low)(1-\mu_0)]$ with $q_i^R<1$ (at $q_i^R=1$ the excluded sale was worthless, the wedge nil); for buyers above, exposure at or above $\bar\rho_i=(q_i^0-q^*)/[(\bar\high-\bar\low)(\mu_0-\pi)]$.
\item[(ii)] \emph{(Asymmetry.)} For $q_i^0\le q^*$ the harm $|W_i|$ is quasi-concave in exposure: zero up to $\underline\rho_i$, then strictly \emph{decreasing}, with supremum $r\low(1-q^*)$ approached at the crossing threshold. For $q_i^0>q^*$ the gain $W_i$ is nondecreasing: zero below $\bar\rho_i$, equal to $(1-r)\low(1-q^*)$ there, strictly increasing beyond. Disclosure harms most the buyers it just carries across the threshold---for whom the destroyed trade was worth most---and helps most those most exposed to the source.
\item[(iii)] \emph{(Redistribution.)} Whenever $F$ carries crossing mass with $q_i^R<1$ below $q^*$ and crossing mass above it, disclosure transfers surplus from exposed buyers below the threshold, whose discounts it destroys, to exposed buyers above it, for whom it creates them---even where the two deadweight-loss changes cancel and $W(g^*;F)=0$.
\end{itemize}
\end{theorem}

In the environment of Example~\ref{ex:single}, a buyer with baseline $q_i^0=5/16$ and exposure $\rho=1/2$ is carried across $q^*=1/3$ by the reveal event and loses $|W_i|=5/32$; doubling her exposure to $\rho=1$ shrinks the harm to $9/64$, because the buyer the market then excludes was less likely to buy at the discount anyway. The harm to the population is concentrated on its marginal members.

\begin{example}[One correlated buyer]\label{ex:ext}
Continue Example~\ref{ex:single}, so $S_1$ discloses with $g_\high=1/2$, revealing $\mu_A=1$ with probability $1/4$ and pooling at $\mu_A=1/3$ with probability $3/4$. Give an exposed buyer $\rho=1/2$ and the lower baseline $\mu_i^0=1/4$, so $q_i^0=5/16<q^*$: absent disclosure $S_2$ discounts her. The reveal event lifts $S_2$'s posterior to $q_i^R=3/8>q^*$ and the buyer loses the discount; the pool event leaves $q_i^P=7/24<q^*$ and she keeps it. Her consumer surplus falls from $(\high-\low)q_i^0=5/8$ under privacy to $7/16$ under disclosure, a loss of $3/16$; $S_2$'s profit on her rises by $1/32$; her buyer-side surplus falls by $5/32$, so $W_i=-5/32=-r\low(1-q_i^R)$---exactly the expected deadweight loss created in the reveal event (Figure~\ref{fig:wedge}). With a mass $m$ of such buyers the aggregate wedge is $-5m/32$ and aggregate consumer surplus falls by $3m/16$.
\end{example}

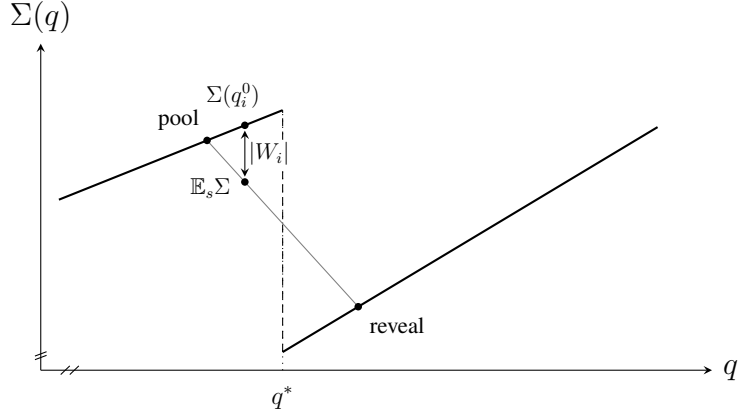
\begin{figure}[t]
\centering
\begin{tikzpicture}[xscale=24,yscale=4.8,>=stealth]
  \def\qs{0.33333}
  % axes (window magnifies a neighborhood of q*; break marks denote the truncated origin)
  \draw[->] (0.20,0.95) -- (0.57,0.95) node[right] {$q$};
  \draw[->] (0.20,0.95) -- (0.20,1.85) node[above] {$\Sigma(q)$};
  \draw (0.211,0.938) -- (0.216,0.962);
  \draw (0.216,0.938) -- (0.221,0.962);
  \draw (0.1968,0.979) -- (0.2032,0.991);
  \draw (0.1968,0.989) -- (0.2032,1.001);
  % the tent: affine branches with a downward jump at q*
  \draw[thick] (0.21,1.42) -- (\qs,1.66667);
  \draw[thick] (\qs,1) -- (0.54,1.62);
  \draw[densely dashed] (\qs,1) -- (\qs,1.66667);
  \draw[dotted] (\qs,0.95) -- (\qs,1.66667);
  \node[below,scale=0.8] at (\qs,0.93) {$q^*$};
  % disclosure spread: chord from pool to reveal
  \draw[gray,thin] (0.29167,1.58333) -- (0.375,1.125);
  \draw[<->,shorten >=2pt,shorten <=2pt] (0.3125,1.625) -- (0.3125,1.46875);
  % points
  \node[fill,circle,inner sep=1pt] at (0.3125,1.625) {};
  \node[fill,circle,inner sep=1pt] at (0.29167,1.58333) {};
  \node[fill,circle,inner sep=1pt] at (0.375,1.125) {};
  \node[fill,circle,inner sep=1pt] at (0.3125,1.46875) {};
  % labels
  \node[above,scale=0.72] at (0.306,1.652) {$\Sigma(q_i^0)$};
  \node[above left,scale=0.72] at (0.2917,1.589) {pool};
  \node[right,scale=0.72,fill=white,inner sep=1pt] at (0.3135,1.547) {$|W_i|$};
  \node[left,scale=0.72] at (0.309,1.452) {$\E_s\Sigma$};
  \node[below right,scale=0.72] at (0.377,1.118) {reveal};
\end{tikzpicture}
\caption{The wedge as a deadweight-loss change (Example~\ref{ex:ext}, with $\high=3$, $\low=1$). Buyer-side surplus $\Sigma$ is affine on each side of $q^*$ with a downward jump $D(q^*{+})=\low(1-q^*)$. Disclosure spreads $S_2$'s belief about the exposed buyer from $q_i^0<q^*$ to a pool and a reveal realization; the reveal realization crosses $q^*$ and forfeits the discount, so $\E_s\Sigma$ lies on the chord, below $\Sigma(q_i^0)$ by $|W_i|$. The plot magnifies a neighborhood of $q^*$; axis breaks mark the truncated origin. Values from the replication certificate.}
\label{fig:wedge}
\end{figure}

\subsection{When is the private rule socially wrong?}

The private rule is $F$-independent; the planner's preferred disclosure is not. The comparison is sharp in the leading case.

\begin{proposition}[Threshold inefficiency]\label{prop:threshold}
In the disclosure region of Theorem~\ref{thm:binary} with $\mu_0>q^*$, let the population place mass $m$ at a single exposed point $(q^0,\rho)$ with $q^0<q^*$, $\rho>\underline\rho$, and reveal posterior $q^R\equiv q^0+\rho(\bar\high-\bar\low)(1-\mu_0)<1$ (automatic under Assumption~\ref{ass:pop}). Every Bayes-plausible signal leaves $S_1$'s allocation at $(q_1(\high),q_1(\low))=(1,0)$, so transfers cancel from welfare and the planner concavifies over posterior support $\{0,\pi,\bar\mu,1\}$, where $\bar\mu=\mu_0+(q^*-q^0)/[\rho(\bar\high-\bar\low)]$ is the largest source posterior still discounting the exposed buyer. The optimum is:
\begin{itemize}[leftmargin=1.6em,itemsep=0.1em,topsep=0.2em]
\item[(i)] the private rule $g^*$ (posteriors $\{\pi,1\}$) if and only if $m\le m^*$ (both rules are optimal at $m=m^*$ exactly), where, with $r=(\mu_0-\pi)/(1-\pi)$,
\[
m^*=\Big[\frac{1-\mu_0}{1-\pi}-\frac{\bar\mu-\mu_0}{\bar\mu-\pi}\Big]\,\frac{1-q^*}{\,r\,(1-q^R)\,};
\]
\item[(ii)] otherwise the \emph{truncated rule} with posteriors $\{\pi,\bar\mu\}$, which preserves the source's discount at zero population harm by capping the reveal posterior at $\bar\mu$.
\end{itemize}
Privacy is never the planner's choice: the truncated rule dominates it by $\frac{\bar\mu-\mu_0}{\bar\mu-\pi}\,\low(1-q^*)>0$ at every $m$. The threshold $m^*$ is increasing in $\rho$ and vanishes as $\rho\downarrow\underline\rho$: the buyers most harmed by disclosure are also the cheapest to protect.
\end{proposition}

The proposition corrects two conjectures at once. The private rule is not ``inefficient whenever it harms someone'': harming a small exposed mass is efficient. And the planner's alternative is not privacy but a \emph{truncation}---disclose enough to win the source her discount, never enough to re-price the exposed population. The privacy debate's binary framing misses the planner's actual margin: the informativeness of the favorable tail. In the environment of Example~\ref{ex:ext}, $m^*=16/15$.

\section{Selling Data, Consent, and Data Minimization}\label{sec:market}

\subsection{A market for data: a neutrality threshold}

So far disclosure has been non-tradeable: $S_1$ discloses only to influence its own trade. We now let $S_1$ sell its signal to $S_2$, capturing a fraction $\alpha\in[0,1]$ of the increase in $S_2$'s expected profit---on the source and on the population---that the signal generates, where $\Pi(q)=\max\{\low,\high q\}$ is $S_2$'s profit on a buyer she believes high with probability $q$. For an exposed buyer write $q_i(\mu)$ for her posterior when the source's is $\mu$, and define her \emph{curvature gap} at the pool,
\[
\sigma_i=(1-\pi)\,\Pi(q_i(0))+\pi\,\Pi(q_i(1))-\Pi(q_i(\pi))\;\ge\;0,
\]
strictly positive exactly when full revelation would carry $S_2$'s belief about her across the pricing threshold, $q_i(0)<q^*<q_i(1)$.

\begin{proposition}[Selling data: a neutrality threshold]\label{prop:sell}
Let the environment lie in the disclosure region of Theorem~\ref{thm:binary}, and define
\begin{align*}
G&=\pi\Big[(\high-\low)\bar\high+\max\{0,\low-\mu_0\high\}\,\frac{\bar\high-\bar\low}{\mu_0}\Big],\\
S&=\pi\high(\bar\high-q^*)+\int\sigma_i\,dF,\qquad \bar\alpha=\frac{G}{S}\,\in(0,\infty).
\end{align*}
\begin{itemize}[leftmargin=1.6em,itemsep=0.1em,topsep=0.2em]
\item[(i)] \emph{(Neutrality.)} For $\alpha<\bar\alpha$ the optimal signal is the rule $g^*$ of Theorem~\ref{thm:binary}, essentially uniquely.
\item[(ii)] \emph{(Bang-bang refinement.)} For $\alpha>\bar\alpha$ the optimal signal is full revelation; no intermediate informativeness arises. (At $\alpha=\bar\alpha$ every signal supported on $\{0,\pi,1\}$ is optimal, and, when no exposed buyer's threshold falls in $(0,\pi)$, every signal supported on $[0,\pi]\cup\{1\}$.)
\item[(iii)] \emph{(Only large markets tip.)} With no exposed population, $G-S=\pi[\low(1-\bar\high)+\max\{0,\low-\mu_0\high\}(\bar\high-\bar\low)/\mu_0]\ge0$, so $\bar\alpha>1$ on the open region. Selling changes disclosure for some $\alpha\in[0,1]$ if and only if $\int\sigma_i\,dF\ge\pi\low(1-\bar\high)+\pi\max\{0,\low-\mu_0\high\}(\bar\high-\bar\low)/\mu_0$. In the environment of Example~\ref{ex:ext} with exposed mass $m$, $\bar\alpha=8/(4+m)$, so the data market changes what is disclosed only once $m\ge4$, with a tie at $m=4$.
\item[(iv)] \emph{(Consumer surplus.)} Aggregate consumer surplus is a step function of $\alpha$: constant on $[0,\bar\alpha)$, one jump at $\bar\alpha$, constant after. The per-buyer jump is $\Delta CS_i=(\high-\low)(1-\mu_0)\big[q_i(0)-q_i(\pi)/(1-\pi)\big]$ on the class $q_i(\pi)\le q^*<q_i(1)$ that the pool discounts but a revealed high source does not, and is strictly negative there; in the environment of Example~\ref{ex:ext} the aggregate jump is $-3m/16$.
\end{itemize}
\end{proposition}

The data market does not refine disclosure at the margin; it tips it. Below the liquidity threshold the right to sell changes how surplus is divided, not what is disclosed---selling is neutral for every population, since $\bar\alpha>0$ always, and the externality of Section~\ref{sec:pop} is the whole story. Above it the market converts the signal into full revelation: the objective is convex in the posterior on each side of $\pi$, so as liquidity crosses $\bar\alpha$ disclosure jumps from the private optimum straight to the maximal element of the Blackwell order, with no intermediate informativeness. The single-buyer neutrality of the bilateral model is exactly the no-population case: the mass whose beliefs full revelation carries across the threshold must be large enough to tip the market, not merely positive. The bilateral intuition that selling is a transfer is exactly the $\int\sigma_i\,dF=0$ case: selling is the channel through which $S_1$ monetizes, and therefore acts on, the cross-buyer informativeness it otherwise ignores.

\subsection{Consent dominates both poles}\label{ss:consent}

The externality calls for an instrument that targets cross-buyer \emph{use}, not the signal. Endow every buyer other than the source with an ex-ante consent right: before the signal realizes, buyer $i$ either consents to or vetoes the use of $S_1$'s signal in $S_2$'s pricing of buyer $i$; a veto obliges $S_2$ to price $i$ at her no-disclosure posterior $q_i^0$. The signal itself, and $S_2$'s pricing of the source, are unrestricted; buyers consent when indifferent; data are non-tradeable ($\alpha=0$).

\begin{proposition}[Consent dominance]\label{prop:consent}
Fix the disclosure region of Theorem~\ref{thm:binary} and the rule $g^*$; buyers consent when indifferent.
\begin{itemize}[leftmargin=1.6em,itemsep=0.1em,topsep=0.2em]
\item[(i)] \emph{(Invariance.)} $S_1$'s problem is unaffected by the consent profile: vetoes restrict only third-party pricing, and consent decisions depend only on the public $(q_i^0,\rho_i)$.
\item[(ii)] \emph{(Alignment.)} For any Bayes-plausible use on buyer $i$, her consumer-surplus change and wedge $W_i$ are never of opposite signs, save two knife-edge atoms: a fully revealed harmed buyer ($q_i^R=1$) vetoes a nil-wedge use, and a benefited buyer revealed to the bottom ($q_i^P=0$) is indifferent and consents. $S_2$'s profit on $i$ rises from use in both. Under $g^*$, a harmed buyer loses consumer surplus $r(\high-\low)q_i^R$ while $S_2$ gains $r\high(q_i^R-q^*)$, a benefited buyer gains $(1-r)(\high-\low)q_i^P$ while $S_2$ gains $(1-r)\high(q^*-q_i^P)$.
\item[(iii)] \emph{(Implementation.)} The equilibrium consent profile implements the planner's optimal buyer-by-buyer use policy,
\[
W_{\mathrm{consent}}=\int\max\{W_i,0\}\,dF=\max_{a:\,(q_i^0,\rho_i)\mapsto\{0,1\}}\int a_i\,W_i\,dF,
\]
and maximizes each buyer's consumer surplus among all use policies for her.
\item[(iv)] \emph{(Dominance.)} $W_{\mathrm{consent}}\ge\max\{W(g^*;F),\,0\}$: consent weakly dominates laissez-faire and a blanket use-ban, which carry identical source-side surplus, strictly over the former iff $F$ has mass on harmed crossing buyers with $q_i^R<1$ and over the latter iff it has mass on benefited crossing buyers.
\end{itemize}
\end{proposition}

The alignment in (ii) is what makes consent informationally cheap: the planner's buyer-by-buyer optimum requires knowing the sign of every $W_i$, and each buyer's self-interest computes it for free. Consent decisions, measurable in the public pair $(q_i^0,\rho_i)$, reveal nothing about the source. Because $S_2$'s profit rises from use even on the buyers the instrument protects, the right must bind her legally; on a benefited buyer at $g^*$ consent is Pareto-improving within the pair. A two-mass illustration in the environment of Example~\ref{ex:ext}: a harmed buyer at $(5/16,1/2)$ with $W_i=-5/32$ and a benefited buyer at $(7/20,1/2)$ with $W_i=161/320$ give laissez-faire $W=111/320$, a use-ban $0$, and consent $161/320$---consent harvests the gain and truncates the harm. The blanket use-ban, the literal reading of data minimization, is the planner's optimum only when no exposed buyer gains; the privacy debate's two poles are both dominated by the instrument that asks each customer.

\subsection{Coarse caps: bang-bang in the binary model, a threshold cap in the continuum}

When neither consent nor buyer-specific use restrictions are feasible, the regulator is left to cap the overall informativeness of the signal. Which caps are bang-bang depends on the margin they move, not on the binary structure: the cap's value turns entirely on whether the source's margin is smooth.

\begin{proposition}[Data-minimization caps]\label{prop:cap}
\begin{itemize}[leftmargin=1.6em,itemsep=0.1em,topsep=0.2em]
\item[(i)] \emph{(Binary model.)} Mixing the rule $g^*$ with silence makes every buyer's posterior lottery---hence the source's value and every wedge---affine in the mixing weight, so welfare is maximized at an endpoint for every population. A cap on the reveal posterior at $\bar\mu$, by contrast, implements the truncated rule of Proposition~\ref{prop:threshold}, the planner's optimum for large exposed mass.
\item[(ii)] \emph{(Continuum: a threshold cap.)} In the continuum model of Section~\ref{sec:engine}, cap the pool's mass at $T$, let $S_1$ re-optimize, and fix the two-message implementation of Lemma~\ref{lem:obed}. The source's value $B(T)$ is strictly increasing up to the unconstrained pool mass $T_{\mathrm{un}}$; the harmed buyer's wedge $W_i(T)$ is zero below the crossing cap $\bar T$, falls discontinuously there, and attenuates beyond. The cap $\bar T$ is welfare-optimal whenever $m\,|W_i(T)|\ge B(T)-B(\bar T)$ for all $T\in(\bar T,T_{\mathrm{un}}]$.
\end{itemize}
\end{proposition}

The discontinuous fall of the wedge at $\bar T$ is exactly the marginal crossing being the most harmful one, by Theorem~\ref{thm:incidence}(ii); the optimal data-minimization rule thus permits the largest disclosure that keeps the exposed population inside its pricing region. In the environment of Example~\ref{ex:engine} with a unit mass of exposed buyers at $(0.30,1)$, the crossing cap is exactly $\bar T=2/3$ against an unconstrained pool mass of $0.74$, the displayed condition holds, and welfare under the cap exceeds both privacy and laissez-faire.

\subsection{What the bilateral optimum does not pin down}

The continuum optimum prescribes each type's discount \emph{odds}, not what the market learns: with obedience binding, every discount message carries posterior $\beta$-mean exactly $q^*$, while the no-discount side may be split into any family of messages with means above $q^*$.

\begin{proposition}[Implementation indeterminacy]\label{prop:implementation}
Under Theorem~\ref{thm:engine}'s conditions with binding obedience: (i) $S_1$'s revenue and every constraint depend on the signal only through $\phi$, so $S_1$ is exactly indifferent across implementations; with binary types the implementation is payoff-unique, so the indeterminacy is a continuum phenomenon. (ii) The choice is payoff-relevant for the exposed population, priced off every message. (iii) Under selling ($\alpha>0$), $S_2$'s willingness to pay selects the convex-order most informative implementation; under consent, vetoing buyers are insulated from it.
\end{proposition}

The bilateral market thus leaves a welfare-relevant margin unpriced. The data market resolves it toward full revelation, which spreads harm to below-threshold buyers the coarser implementation spared while diluting the harm of those it crossed: no implementation is uniformly consumer-preferred, and the margin is genuinely policy-relevant.
\section{Robustness}\label{sec:robust}

Two features of the model invite scrutiny: the sequential resolution of the buyer's uncertainty and the monopoly structure of the downstream market. The first is essential and identifies the paper's engine; the second is a comparative static, not a hostage.

\subsection{Sequential learning is essential}

In the baseline model the buyer learns his second-period valuation only when he meets $S_2$. Suppose instead he knows both valuations when he contracts with $S_1$, so his type is the pair $(\theta_1,\theta_2)\in\{\low,\high\}^2$. This is the static environment of \citet{CalzolariPavan2006}.

\begin{proposition}[Sequential learning is essential]\label{prop:static}
Suppose the buyer knows both $\theta_1$ and $\theta_2$ when contracting with $S_1$, and let preferences be perfectly positively correlated or let $\bar\low\ge q^*$. Then privacy is optimal: no informative signal raises $S_1$'s revenue.
\end{proposition}

\begin{remark}[The full static region]\label{rem:staticregion}
The conclusion is not special to the cases the proposition proves. With both valuations known, the seller's problem reduces to a four-type direct mechanism with two aggregate obedience constraints (the reduction is proved with Proposition~\ref{prop:static}), a parametric linear program. On the remaining region $\bar\low<q^*<\bar\high$ of
\begin{equation}\label{eq:staticcond}
  \max\{\bar\low,\ \low/\high\}<\bar\high ,
\end{equation}
the same conclusion is a computational finding rather than a theorem: linear programming on a covering grid and at four thousand interior draws finds no informative signal improving on privacy (replication certificate). The condition cannot be dropped: outside it, a signal conditioning on the second-period taste can strictly raise revenue---the proof of Proposition~\ref{prop:static} exhibits a mechanism earning $1.8$ against the privacy revenue $1.5$.
\end{remark}

The contrast with Theorem~\ref{thm:binary} is sharp. With sequential learning, disclosure is optimal on an open set of imperfect correlations; with static types, privacy prevails throughout \eqref{eq:staticcond}. What disclosure sells is the prospect of a downstream discount; that prospect has value only to a buyer who does not yet know whether he will want the second good. Resolve his uncertainty and the instrument loses its purchase. The robustness of the \citet{CalzolariPavan2006} privacy result is thus a feature of their static types, not of correlation \emph{per se}: it is exactly the stochastic evolution of preferences, in the sell-the-option sense of \citet{EsoSzentes2007}, that overturns it.

\subsection{Competition: the externality is a market-power phenomenon}

The wedge of Proposition~\ref{prop:dwl} operates on the downstream deadweight loss. Competition disciplines that loss, and with it the externality, in a way the representation makes exact. Let buyer-side surplus count a buyer's consumer surplus and all downstream producer surplus on her.

\begin{proposition}[Downstream competition]\label{prop:competition}
Modify the downstream market, holding fixed the population structure of Assumption~\ref{ass:pop}.
\begin{itemize}[leftmargin=1.6em,itemsep=0.1em,topsep=0.2em]
\item[(i)] \emph{(Bertrand.)} If two or more downstream sellers with identical zero marginal cost post prices, the equilibrium price is zero after every belief: every type trades, $D\equiv0$, and $W_i(g)=0$ for every buyer and signal. Beliefs do not move price, so the discount premium and $S_1$'s gain from disclosure are zero---under competition the externality has nothing to act on.
\item[(ii)] \emph{(Constrained monopoly.)} Let $S_2$ face an outside option at price $p^{\mathrm o}>\low$ for a perfect substitute, and write $\tilde p=\min\{\high,p^{\mathrm o}\}\in(\low,\high]$ ($\tilde p=\high$ the monopoly benchmark). Then $S_2$ discounts belief $q$ if and only if $q\le q^*_{\mathrm o}\equiv\low/\tilde p$, the deadweight loss is the tent $D_{\mathrm o}(q)=\low(1-q)\ind\{q>q^*_{\mathrm o}\}$, and for every Bayes-plausible $g$, $W_i(g)\le0$ when $q_i^0\le q^*_{\mathrm o}$ and $W_i(g)\ge0$ when $q_i^0>q^*_{\mathrm o}$, with the asymmetry of Theorem~\ref{thm:incidence} verbatim at $q^*_{\mathrm o}$ and $|W_i(g)|\le\low(1-\low/\tilde p)$, tight and falling to zero as $\tilde p\downarrow\low$. The discount premium $\tilde p-\low$, and with it $S_1$'s gain from disclosure, vanishes at the same rate. Once $q(\mu_0)\le q^*_{\mathrm o}$ the market discounts absent any disclosure, no informative signal raises $S_1$'s revenue, and---$S_1$ choosing silence when indifferent---the equilibrium externality is identically zero.
\item[(iii)] \emph{(The threshold moves; the law does not.)} As $\tilde p$ falls, the threshold $\low/\tilde p$ rises, and a buyer with $q_i^0\in(\low/\high,\,\low/\tilde p]$---a beneficiary under monopoly---becomes weakly harmed: competition itself supplies the discount disclosure created for her, leaving disclosure only the power to destroy it. The same law governs at every $\tilde p>\low$---harm below the prevailing threshold, gain above---only the assignment of buyers to the two sides shifting.
\end{itemize}
\end{proposition}

In the environment of Example~\ref{ex:ext} the externality vanishes already at $\tilde p\le8/3$. Downstream competition therefore neither reverses the externality's structure nor preserves its incidence: it relocates the threshold and shrinks every buyer's stake, uniformly, at the rate of the downstream margin. The privacy externality is a phenomenon of downstream market power, largest where pricing is least contested---which is where the policy instruments of Section~\ref{sec:market} have their work to do.

\section{Conclusion}\label{sec:conclusion}

This paper has studied a firm's incentive to disclose or sell data about its customers, and traced its consequences once those customers' data are correlated. Disclosure is a costly screening instrument: the firm reveals a customer's record to persuade a downstream seller to discount and charges the customer for the prospect. The characterization of when the instrument earns its keep is exact, and incentive compatibility shapes it twice over---a rent to the type who covets the discount, a supply distortion or an exclusion to the type who would chase it, and, with a continuum of types, a rationing of the discount's odds at the top of the type space, where deterministic disclosure is strictly suboptimal. The classical optimality of privacy is the perfectly correlated boundary of the disclosure region, and what overturns it is the sequential resolution of the buyer's tastes.

Correlation changes whom the instrument touches. Disclosure about one customer moves the market's belief about every similar customer---and the price of every customer it carries across the pricing threshold---a privacy externality the firm does not internalize, equal to the disclosure-induced reduction in downstream deadweight loss. Its incidence is asymmetric---concentrated on the buyers it just carries across the pricing threshold---and its policy economics are not the privacy debate's binary. The planner's instrument is a truncation, not a ban; consent dominates both poles because each customer's self-interest computes the sign of her own wedge; caps that dilute the signal are bang-bang while caps on its favorable tail implement the truncation, the continuum optimum stopping at the population's crossing threshold; and the data market is neutral until an explicit liquidity threshold, past which it tips disclosure to full revelation. The externality itself is a creature of downstream market power, vanishing under Bertrand pricing at the same rate as the instrument's private value.

A record disclosed today prices transactions tomorrow: persistent records compound the externality across time, and the buyer's anticipation of future disclosure feeds back into today's reports. The costly-screening language developed here---an instrument priced by the monotonicity it must respect, an externality summarized by the deadweight loss it moves---prices a single record's disclosure, and the same accounting extends across periods once records persist.

\appendix
\section{Proofs for Sections 3--4}\label{app:binary}

\noindent\textbf{Proof of Lemma~\ref{lem:cont}.} Facing a buyer with $\Prob(\theta_2=\high)=q$, seller $S_2$ earns $\high q$ from posting $\high$ (only the high type buys) and $\low$ from posting $\low$ (both buy), so she discounts iff $\low\ge\high q$, i.e.\ $q\le q^*$, equivalently $\mu_1\le\pi$ by \eqref{eq:pi}, the tie at $q=q^*$ resolved toward the discount by the convention of Section~\ref{sec:model}. At the discount price $\low$ a high second-period type gets $\high-\low$ and a low type $0$; at $\high$ both get $0$. Hence the high type's surplus is $Q(\mu_1)=(\high-\low)\ind\{\mu_1\le\pi\}$ and the low type's is $0$. Conditioning on a report $b_1$ and integrating over the signal gives the stated continuation values. $\square$

\medskip
\noindent\textbf{Proof of the feasibility aggregation.} We record once, for both correlation signs, the feasible set of discount probabilities $(\phi_\low,\phi_\high)$. Under positive correlation with $\mu_0>\pi$, a message $s$ triggers the discount iff its posterior satisfies $\mu_1(s)\le\pi$. For any signal, summing Bayes' rule over the discount messages gives $\mu_0(1-\pi)\phi_\high\le\pi(1-\mu_0)\phi_\low$, i.e.\ $\phi_\high\le k\phi_\low$; richer message structures cannot enlarge the set, since the bound aggregates message-level constraints. Conversely, any pair with $\phi_\high\le k\phi_\low$ is implemented by pooling all discount weight into a single message, whose posterior is at most $\pi$ by construction, the complement then carrying posterior above $\pi$ because the prior $\mu_0>\pi$ is their convex combination (the complement is nonempty: $\phi_\high=\phi_\low=1$ would put the prior in the discount message, contradicting $\mu_0>\pi$). The negative-correlation case is the mirror image: discount iff $\mu_1\ge\pi$, and the feasible set is $\phi_\low\le\phi_\high/\tilde k$ when $\mu_0<\pi$. $\square$

\medskip
\noindent\textbf{Proof of Proposition~\ref{prop:fb}.} \emph{Reduction.} With $\theta_1$ observable, type $\theta_1$ accepts iff $\theta_1 q_1(\theta_1)-p_1(\theta_1)+\beta(\theta_1)\E_g[Q\mid\theta_1]\ge0$, where $\beta(\high)=\bar\high$, $\beta(\low)=\bar\low$. Both constraints bind, else $S_1$ raises a price; substituting the transfers into revenue gives \eqref{eq:fbobj}, whose trade term is maximized at $q_1\equiv1$. \emph{Concavification.} For a signal with unconditional probabilities $\bar g(s)$ and posteriors $\mu_1(s)$, Bayes' rule turns the disclosure term into $\E_{\bar g}[J_f(\mu_1)]$ with $J_f$ as in \eqref{eq:Jf}. By \citet{KamenicaGentzkow2011}, $S_1$ may induce any Bayes-plausible distribution of posteriors, so the optimal value is the concave closure of $J_f$ at $\mu_0$ and an informative signal strictly dominates iff the closure strictly exceeds $J_f(\mu_0)$. \emph{Region.} On $[0,\pi]$, $J_f$ rises from $(\high-\low)\bar\low$ to $(\high-\low)q^*$ and drops to $0$ on $(\pi,1]$. If $\bar\low>q^*$ then $\pi<0$ and $J_f\equiv0$: disclosure is worthless (the slice $\bar\low=q^*$ is excluded by hypothesis and settled in Remark~\ref{rem:boundary}). If $\bar\high=\bar\low$ every posterior carries the same $q$ and disclosure is payoff-irrelevant. If $\bar\high>\bar\low$, $\bar\low<q^*$, and $\mu_0\le\pi$ the prior already discounts and $J_f$ is concave at $\mu_0$. If $\bar\low<q^*$ and $\mu_0>\pi$, the chord from $(\pi,(\high-\low)q^*)$ to $(1,0)$ lies strictly above $J_f(\mu_0)=0$, so disclosure strictly dominates; the optimal posteriors $\{\pi,1\}$ are implemented by \eqref{eq:optrule}. $\square$

\medskip
\noindent\textbf{Proof of Theorem~\ref{thm:binary}.} Write $\Delta=\high-\low$, $u_b$ for the truthful payoff of type $b$, and $\phi_b$ for the discount probability of report $b$. The two incentive constraints read
\[
u_\high-u_\low\;\ge\;\Delta\,q_1(\low)+\Delta(\bar\high-\bar\low)\,\phi_\low
\;=:\;A_\low,
\qquad
u_\high-u_\low\;\le\;\Delta\,q_1(\high)+\Delta(\bar\high-\bar\low)\,\phi_\high
\;=:\;A_\high .
\]
\emph{Step 1 (feasibility and rents).} The pair is feasible iff $A_\low\le A_\high$, which is \eqref{eq:M}; if \eqref{eq:M} fails the constraint set is empty. Revenue is decreasing in $(u_\low,u_\high)$, so at the optimum $u_\low=0$ and $u_\high=A_\low\ge0$: IR$_\low$ and IC$_\high$ bind, and IC$_\low$ binds exactly when \eqref{eq:M} binds. Substituting,
\begin{equation}\label{eq:reduced}
R=\mu_0\high\,q_1(\high)+(\low-\mu_0\high)\,q_1(\low)+\Delta\big[\mu_0\bar\high\,\phi_\high+(\bar\low-\mu_0\bar\high)\,\phi_\low\big],
\end{equation}
to be maximized subject to \eqref{eq:M}, $\phi_\high\le k\phi_\low$, and box constraints. Since lotteries enter linearly and the feasibility aggregation above is already stated for arbitrary signals, no stochastic menu enlarges the problem.

\emph{Step 2 (dominant margins).} Raising $q_1(\high)$ raises $R$ and relaxes \eqref{eq:M}: $q_1(\high)=1$. Raising $\phi_\high$ to its persuasion bound $k\phi_\low$ raises $R$ (coefficient $\mu_0\bar\high\Delta>0$) and relaxes \eqref{eq:M}: $\phi_\high=k\phi_\low$.

\emph{Step 3 (case $\mu_0\ge q^*$).} The coefficient on $q_1(\low)$ is $\low-\mu_0\high\le0$, so $q_1(\low)=0$, and \eqref{eq:M} reads $1+(\bar\high-\bar\low)k\phi_\low\ge(\bar\high-\bar\low)\phi_\low$, which holds strictly for every $\phi_\low\le1$: IC$_\low$ is slack. $R$ is then linear in $\phi_\low$ with slope $\Delta[\bar\low-\mu_0\bar\high g_\high]$, where we used $1-k=g_\high$, and the algebraic identity
\[
\bar\low-\mu_0\bar\high\,g_\high=\frac{(\bar\high-\bar\low)(q^*-\mu_0\bar\high)}{\bar\high-q^*},
\qquad g_\high=\frac{q(\mu_0)-q^*}{\mu_0(\bar\high-q^*)},
\]
both verified by clearing denominators, signs the slope by $q^*-\mu_0\bar\high$. Disclosure ($\phi_\low=1$) is strictly optimal iff $\mu_0\bar\high<q^*$, with rent $u_\high=\Delta(\bar\high-\bar\low)$.

\emph{Step 4 (case $\mu_0<q^*$).} Now $\low-\mu_0\high>0$, so $q_1(\low)$ rises to its \eqref{eq:M}-cap $1-(\bar\high-\bar\low)g_\high\phi_\low$. Substituting,
\[
R(\phi_\low)=\low+\phi_\low\Big\{\Delta\big[\bar\low-\mu_0\bar\high g_\high\big]-(\low-\mu_0\high)(\bar\high-\bar\low)g_\high\Big\},
\]
linear in $\phi_\low$, so disclosure is strictly optimal iff the brace is positive, which clearing $g_\high$ is the left side of \eqref{eq:identity} being positive. The identity \eqref{eq:identity} itself is a polynomial identity in $(\mu_0,q^*,\bar\high,\bar\low)$ after substituting $\low=\high q^*$ and dividing by $\high$; expanding both sides verifies it term by term. On the open region \eqref{eq:sbregion} with $\mu_0<q^*$ the first bracketed term on the right of \eqref{eq:identity} is nonnegative and the second strictly positive ($\mu_0<q^*$, $\bar\low<q^*$), so the brace is strictly positive: disclosure obtains everywhere there, the deterrence cost notwithstanding, and $\mu_0\bar\high\le\mu_0<q^*$ shows \eqref{eq:sbregion}'s first inequality is automatic on this case. The mechanism reads off the binding constraints: $\phi_\low=1$, $\phi_\high=k$, $q_1(\low)=1-(\bar\high-\bar\low)g_\high$, $u_\high=\Delta[q_1(\low)+(\bar\high-\bar\low)]$, both incentive constraints binding.

\emph{Step 5 (outside the region).} If $\bar\low>q^*$ then $\pi<0$: no posterior induces a discount, $\phi\equiv0$, and the problem is the standard one-good monopoly, i.e.\ privacy. If $q(\mu_0)\le q^*$ the prior itself is discounted. When $\bar\high\le q^*$ every posterior discounts, so $\phi_\high=\phi_\low=1$ is forced and every signal is payoff-equivalent to privacy. Otherwise $\pi\in[\mu_0,1)$, and summing the message-level no-discount conditions shows the feasible set is $\{\phi_\high=\phi_\low=1\}$ together with pairs satisfying $\mu_0(1-\pi)(1-\phi_\high)\ge\pi(1-\mu_0)(1-\phi_\low)$, that is, $1-\phi_\high\ge\tilde k(1-\phi_\low)$ with $\tilde k=\pi(1-\mu_0)/[\mu_0(1-\pi)]\ge1$ here (the formula of Section 4.1's $k$, on the other side of $\pi$); along the binding boundary the revenue slope in $\phi_\low$ is proportional to $\bar\low+\mu_0\bar\high(\tilde k-1)\ge0$, so revenue weakly increases toward the uninformative point $(1,1)$: privacy (with the discount priced into both contracts) is optimal and no informative signal improves on it. If $\bar\low<q^*<q(\mu_0)$ but $\mu_0\bar\high\ge q^*$, Step 3's slope is nonpositive and privacy is weakly, and for $\mu_0\bar\high>q^*$ strictly, optimal; at $\mu_0\bar\high=q^*$ exactly, discount-inducing signals tie with privacy.

\emph{Step 6 (essential uniqueness).} On the open region with $\mu_0\neq q^*$: the slopes in Steps 3--4 are strictly signed, so $\phi_\low=1$, $\phi_\high=k\phi_\low$, and $q_1(\low)$ (at its stated corner or cap) are unique, as are the rents. For the signal: given $\phi_\low=1$ and $\phi_\high=k$, the discount event's aggregate posterior is $\mu_0k/(\mu_0k+1-\mu_0)=\pi$; each discount message carries posterior at most $\pi$ and they average $\pi$, so every discount message carries posterior exactly $\pi$; non-discount messages have zero likelihood under $\low$, hence posterior $1$. The signal is the two-point rule \eqref{eq:optrule}, uniquely up to relabeling and duplication of messages. At $\mu_0=q^*$ exactly, the coefficient on $q_1(\low)$ vanishes and a continuum of optimal supply levels coexists---the knife edge is genuine---while the $\phi$-slopes remain strictly signed, so the signal stays essentially unique. $\square$

\medskip
\noindent\textbf{Proof of Corollary~\ref{cor:cp}.} At $(\bar\high,\bar\low)=(1,0)$, $q(\mu_0)=\mu_0=\mu_0\bar\high$, so \eqref{eq:sbregion} requires $\mu_0\bar\high<q^*<q(\mu_0)=\mu_0\bar\high$, which is empty; by Theorem~\ref{thm:binary}'s Step 5, privacy is optimal. $\square$

\medskip
\noindent\textbf{Proof of Remark~\ref{rem:boundary}.} At $\bar\low=q^*$, $\pi=0$: the only discount message is the full revelation of $\low$, so $\phi_\high=0$ and $\phi_\low\in[0,1]$ is free. In the first best, the disclosure value per unit of $\phi_\low$ is $(1-\mu_0)\bar\low\Delta>0$, so disclosure is strictly optimal whenever a discount is worth manufacturing, $q(\mu_0)>q^*$. In the second best under $\mu_0\ge q^*$: $q_1(\low)=0$ and the $\phi_\low$-slope is the Step-3 expression $\Delta[\bar\low-\mu_0\bar\high g_\high]$ evaluated at $g_\high=1$, namely $\Delta(q^*-\mu_0\bar\high)$, positive iff $\mu_0\bar\high<q^*$. Under $\mu_0<q^*$: the \eqref{eq:M}-cap gives $q_1(\low)=1-(\bar\high-\bar\low)\phi_\low$, and the slope works out to $\high\,q^*(1-\bar\high)(1-\mu_0)$, strictly positive iff $\bar\high<1$. $\square$

\medskip
\noindent\textbf{Proof of Proposition~\ref{prop:negcorr}.} With $\bar\high<\bar\low$, $q(\mu_1)$ is strictly decreasing, $S_2$ discounts iff $\mu_1\ge\pi$, and feasibility aggregates to $\phi_\low\le\phi_\high/\tilde k$ on the regime $\mu_0<\pi$. The incentive constraints become
\[
A_\low=\Delta q_1(\low)-\Delta(\bar\low-\bar\high)\phi_\low\;\le\;u_\high-u_\low\;\le\;\Delta q_1(\high)-\Delta(\bar\low-\bar\high)\phi_\high=A_\high,
\]
feasible iff $A_\low\le A_\high$. Rent minimization gives $u_\low=0$, $u_\high=\max\{0,A_\low\}$ (when $A_\low<0\le A_\high$ both rents are zero; $A_\high<0$ is impossible at $q_1(\high)=1$ since $(\bar\low-\bar\high)\phi_\high\le1$). As before $q_1(\high)=1$. The optimal $\phi_\high$ is settled by case below; in both cases $\phi_\low$ sits at its persuasion bound $\phi_\high/\tilde k$, since its revenue coefficient $(1-\mu_0)[\low(\bar\low-\bar\high)+\Delta\bar\low]$ (case (i), through the zero-rent cap on $q_1(\low)$) or $(1-\mu_0)\Delta\bar\low$ (case (ii)) is strictly positive.

\emph{Case (i): $\mu_0\ge q^*$.} The unconstrained coefficient on $q_1(\low)$ in revenue is $(1-\mu_0)\low>0$ \emph{as long as no rent is paid}; the zero-rent condition is $A_\low\le0$, i.e.\ $q_1(\low)\le(\bar\low-\bar\high)\phi_\low$. Raising $q_1(\low)$ beyond the cap triggers the rent $\mu_0\,\Delta$ per unit, which exceeds the revenue gain $(1-\mu_0)\low$ iff $\mu_0\high>\low$, i.e.\ iff $\mu_0>q^*$; at $\mu_0=q^*$ rent and gain are equal, the tie recorded in the statement. In case (i), $\phi_\high=1$ is costless---feasibility $A_\low\le A_\high$ is slack at $q_1(\low)$ below the zero-rent cap, and $\phi_\high$ relaxes the persuasion bound while weakly raising revenue. So $q_1(\low)=(\bar\low-\bar\high)/\tilde k$, both rents zero, and revenue under disclosure is
\[
R=\mu_0[\high+\Delta\bar\high]+(1-\mu_0)\Big[\low\,\frac{\bar\low-\bar\high}{\tilde k}+\Delta\,\frac{\bar\low}{\tilde k}\Big]
\;>\;\mu_0\high=R_{\mathrm{privacy}},
\]
strictly, since the added terms are nonnegative and the low-type terms $(1-\mu_0)[\low(\bar\low-\bar\high)+\Delta\bar\low]/\tilde k$ are strictly positive ($\bar\low>q^*>0$ on the regime): disclosure is strictly optimal everywhere on the regime. Given $(\phi_\high,\phi_\low)=(1,1/\tilde k)$, the discount event's aggregate posterior is $\pi$; each discount message carries posterior at least $\pi$ and they average $\pi$, so all equal $\pi$, and non-discount messages have zero likelihood under the high report, hence posterior $0$: the signal is the stated reveal-low rule, essentially uniquely.

\emph{Case (ii): $\mu_0<q^*$.} Now $(1-\mu_0)\low>\mu_0\Delta$, so serving the low type fully is worth the rent: fixing any $\phi_\high=s\in[0,1]$ and $\phi_\low=s/\tilde k$, the candidate optima are the rent-paying corner---$q_1(\low)$ at its feasibility cap $1-(\bar\low-\bar\high)(1-1/\tilde k)s$, since $A_\low\le A_\high$ requires $q_1(\low)\le1-(\bar\low-\bar\high)(\phi_\high-\phi_\low)$---and the zero-rent corner $q_1(\low)=(\bar\low-\bar\high)s/\tilde k$ of case (i). Direct comparison at any $s$: rent-paying minus zero-rent equals $(\low-\mu_0\high)\,[1-(\bar\low-\bar\high)\,s]$, decreasing in $s$ and hence minimized at $s=1$, where it is $(\low-\mu_0\high)[1-(\bar\low-\bar\high)]\ge0$, strict unless $\bar\low-\bar\high=1$, where the two mechanisms coincide; a fortiori both dominate excluding the low type, the rent-paying corner at $s=1$ exceeding exclusion by $q_1(\low)(\low-\mu_0\high)+\mu_0\Delta(\bar\low-\bar\high)/\tilde k>0$. Along the rent-paying corner, revenue is linear in $s$ with slope exactly the difference of the two sides of \eqref{eq:negB}, so the optimum is $s=1$ (disclosure) iff \eqref{eq:negB} holds and $s=0$ (privacy) otherwise. The disclosure mechanism's rent is $u_\high=\Delta[q_1(\low)-(\bar\low-\bar\high)/\tilde k]=\Delta[1-(\bar\low-\bar\high)]$, and against privacy ($R_{\mathrm{privacy}}=\low$, the pooling price, with rent $\Delta$): disclosure strictly optimal iff \eqref{eq:negB}. The inequality fails on an open subset of the regime (e.g.\ $\bar\high$ near $q^*$, $\tilde k$ large), and holds on another (e.g.\ $\mu_0$ near $q^*$), so both outcomes have positive measure. $\square$

\medskip
\noindent\textbf{Proof of Corollary~\ref{cor:welfare}.} IR$_\low$ binds in every case, so the low type is indifferent. Positive correlation: under privacy the high type's rent is $\Delta q_1^{\mathrm{priv}}(\low)$, equal to $0$ when $\mu_0\ge q^*$ (the low type is excluded) and $\Delta$ when $\mu_0<q^*$ (pooling at price $\low$). Under disclosure his rent is $u_\high=\Delta[q_1(\low)+(\bar\high-\bar\low)\phi_\low]$: case (i), $\Delta(\bar\high-\bar\low)>0$; case (ii), $\Delta[1-(\bar\high-\bar\low)g_\high+(\bar\high-\bar\low)]=\Delta[1+(\bar\high-\bar\low)(1-g_\high)]>\Delta$. He strictly gains in both. Negative correlation: case (i), both rents are zero under disclosure and under privacy ($q_1^{\mathrm{priv}}(\low)=0$ at $\mu_0\ge q^*$): indifferent. Case (ii), privacy rent $\Delta$ versus disclosure rent $\Delta[1-(\bar\low-\bar\high)]$: he loses exactly $\Delta(\bar\low-\bar\high)$. $\square$
\section{Proofs for Section 5}\label{app:engine}

\noindent\textbf{Proof of Lemma~\ref{lem:zred}.} Under Condition~\ref{cond:le}, the payoff of type $\theta$ from the bundle of report $b$ is $\theta z(b)+\big[(\high-\low)\beta_0\phi(b)-t(b)\big]$: preferences over reports are quasi-linear in the adjusted transfer $\tau(b)=t(b)-(\high-\low)\beta_0\phi(b)$ and linear in type with allocation $z(b)\in[0,1+c]$. The menu $\{(z(b),\tau(b))\}$ is therefore a one-dimensional screening menu, and the standard characterization \citep{Myerson1981,MilgromSegal2002} applies verbatim: incentive compatibility holds if and only if $z$ is nondecreasing in the report and $\tau$ satisfies the envelope formula $U(\theta)=U(0)+\int_0^\theta z(a)\,da$ with $U(\theta)=\theta z(\theta)-\tau(\theta)$. Translating back through $\tau$, the transfer formula is $t(\theta)=\theta z(\theta)+(\high-\low)\beta_0\phi(\theta)-U(\theta)$. $\square$

\begin{lemma}[Obedience aggregation]\label{lem:obed}
On the regime $\E[\beta]>q^*$, a discount-probability profile $\phi:[0,1]\to[0,1]$ is induced by some signal whose every discount message is obeyed by $S_2$ if and only if $\E[(\beta(\theta)-q^*)\,\phi(\theta)]\le0$.
\end{lemma}

\noindent\textbf{Proof.} Necessity: each discount message $s$ satisfies $\E[\beta\mid s]\le q^*$, i.e.\ $\int(\beta-q^*)g(s\mid\theta)\,dF\le0$; summing over discount messages yields the aggregate inequality. Sufficiency: if $\phi=0$ a.e.\ the uninformative signal works; otherwise pool the entire discount weight into a single message $d$ with $g(d\mid\theta)=\phi(\theta)$; its posterior $\beta$-mean is at most $q^*$ by the aggregate inequality, so $S_2$ discounts at $d$ (with equality, by the tie-breaking convention). The complementary message $n$ has $\beta$-mean $\big(\E[\beta]-\E[\beta\phi]\big)/\E[1-\phi]>q^*$: its numerator exceeds $q^*\E[1-\phi]$ because $\E[(\beta-q^*)(1-\phi)]=\E[\beta-q^*]-\E[(\beta-q^*)\phi]\ge\E[\beta-q^*]>0$; and $\E[1-\phi]>0$, since $\phi\equiv1$ would give $\E[(\beta-q^*)\phi]=\E[\beta-q^*]>0$, contradicting the aggregate inequality. So $S_2$ posts $\high$ at $n$, and the two-message signal implements $\phi$. $\square$

\medskip
\noindent\textbf{Proof of the revenue representation \eqref{eq:revrep}.} Let $U(\theta)$ be the equilibrium payoff. By Lemma~\ref{lem:zred}, $U'(\theta)=z(\theta)=x(\theta)+(\high-\low)\beta'(\theta)\phi(\theta)$ a.e.\ (using $\beta'=\beta_1$), with $U(0)=0$ at the optimum (a positive intercept is a uniform transfer loss). Truthful payoffs give $t(\theta)=x(\theta)\theta+(\high-\low)\beta(\theta)\phi(\theta)-U(\theta)$, so $\E[t]=\E[x\theta+(\high-\low)\beta\phi]-\E[U]$. By Fubini, $\E[U]=\int_0^1U'(a)(1-F(a))\,da$. Substituting and collecting terms under the integral yields $\E[t]=\E[x\,\psi]+(\high-\low)\E[\phi\,v]$ with $\psi$, $v$ as in \eqref{eq:revrep}. $\square$

\medskip
\noindent\textbf{Proof of Theorem~\ref{thm:engine}.} Throughout, $\Delta=\high-\low$, $c=\Delta\beta_1$, and the regime $\beta(0)<q^*<\E[\beta]$ holds. By Lemma~\ref{lem:zred}, the representation, and Lemma~\ref{lem:obed}, the second best is exactly the program \eqref{eq:program}. The proof is in four claims; it is self-contained, and the replication certificate accompanying the paper independently corroborates it, solving the full direct-incentive linear program (all pairwise constraints) on fine type grids and reproducing \eqref{eq:program}'s value to machine precision alongside every numerical magnitude below.

\emph{Claim 1 (Lagrangian sufficiency).} Let $\lambda^*\ge0$, let $(x^\circ,\phi^\circ)$ be feasible for \eqref{eq:program} with $\E[(\beta-q^*)\phi^\circ]=0$ if $\lambda^*>0$, and suppose $(x^\circ,\phi^\circ)$ maximizes the Lagrangian $\E[x\psi+\phi\,w]$, $w=\Delta v-\lambda^*(\beta-q^*)$, over all $(x,\phi)\in[0,1]^2$ with $z=x+c\phi$ nondecreasing. Then $(x^\circ,\phi^\circ)$ solves \eqref{eq:program}. Indeed, for any feasible $(x,\phi)$, its Lagrangian value is at least its objective value (obedience $\le0$, $\lambda^*\ge0$), is at most the Lagrangian value of $(x^\circ,\phi^\circ)$, which equals its objective value by complementary slackness.

\emph{Claim 2 (inner split and pointwise value).} For fixed $z$, maximizing $x\psi+\phi w$ subject to $x+c\phi=z$, $x,\phi\in[0,1]$ is linear in $\phi$ over the interval $[\max\{0,(z-1)/c\},\min\{1,z/c\}]$ with slope $s_{\lambda^*}=w-c\psi=\Delta\beta_0-\lambda^*(\beta-q^*)$ by the constant-drift identity. So $\phi$ sits at the upper bound where $s_{\lambda^*}>0$ and the lower where $s_{\lambda^*}<0$, and the resulting pointwise value $V(z;\theta)$ is concave and piecewise linear in $z$, with slopes (writing $\tilde\gamma=\psi+s_{\lambda^*}/c$): for $s_{\lambda^*}>0$, $\tilde\gamma$ on $[0,c]$ then $\psi$ on $[c,1+c]$; for $s_{\lambda^*}<0$, $\psi$ on $[0,1]$ then $\tilde\gamma$ on $[1,1+c]$.

\emph{Claim 3 (construction and verification, $F$ uniform, $c\le1$).} Under uniform $F$, $\psi(\theta)=2\theta-1$ and $\tilde\gamma'=2-\lambda^*/\Delta$. First, the plateau condition collapses: substituting $s_{\lambda}(\theta_\lambda)=0$ and integrating, $\int_{\theta^*}^{1}\psi\,d\theta+\tfrac1c\int_{\theta_{\lambda}}^{1}s_{\lambda}\,d\theta=\tfrac14-\lambda(1-\theta_\lambda)^2/(2\Delta)$, so the condition reads $(1-\theta_\lambda)^2=\Delta/(2\lambda)$; in particular $\theta_{\lambda^*}>\theta^*=\tfrac12$ if and only if $\lambda^*>2\Delta$, so the conflict case carries $\lambda^*>2\Delta$ and $\tilde\gamma$ strictly decreasing. Active disclosure requires $\tilde\gamma(0)>0$ (else $\phi\equiv0$), and the case condition $\theta^*<\theta_{\lambda^*}$ is equivalent to $\tilde\gamma(\theta^*)=s_{\lambda^*}(\theta^*)/c>0$, so $\tilde\gamma>0$ on $[0,\theta^*]$. The pointwise maximizer is then $z^*=c$ on $[0,\theta^*)$ ($\psi<0<\tilde\gamma$), $z^*=1+c$ on $[\theta^*,\hat\theta_\gamma)$ where $\tilde\gamma$ crosses zero, and $z^*=1$ on $[\hat\theta_\gamma,1]$: one decreasing step. The candidate optimum irons it: $z^\circ=c$ on $[0,\theta^*)$ and $z^\circ=\bar z$ on $[\theta^*,1]$, with $\phi^\circ$ read off Claim 2 ($\phi^\circ=1$ below $\theta_{\lambda^*}$ since $\bar z\ge1\ge c$; $\phi^\circ=(\bar z-1)/c=\bar\phi$ above) and $x^\circ=z^\circ-c\phi^\circ$ as displayed in the theorem. Choose the supergradient selection $\gamma(\theta)=0$ on $[0,\theta^*)$ (admissible at the kink $z=c$, where the supergradient is $[\psi,\tilde\gamma]\ni0$), $\gamma=\psi$ on $[\theta^*,\theta_{\lambda^*}]$ and $\gamma=\tilde\gamma$ on $(\theta_{\lambda^*},1]$ (both forced, $V$ being differentiable at $\bar z\in(1,1+c)$ on the plateau), and let $(\lambda^*,\bar\phi)$ solve the system of the theorem; for $\bar\phi>0$ both selections on the plateau are forced, $V$ being differentiable at $\bar z\in(1,1+c)$, while the corner $\bar\phi=0$ uses the kink selection of the corner argument below. Then with $M(\hat\theta)=\int_{\hat\theta}^1\gamma\,dF$: $M(\theta^*)=M(1)=0$; $\gamma$ rises from $0$ on $[\theta^*,\theta_{\lambda^*}]$, is continuous at $\theta_{\lambda^*}$, and is strictly decreasing after, crossing zero once at $\hat\theta_\gamma$; hence $M\le0$ on $[\theta^*,1]$, with strict inequality on the open interior. For any feasible nondecreasing $z$, concavity gives
\[
\int\!\big[V(z;\theta)-V(z^\circ;\theta)\big]dF\;\le\;\int\!\gamma\,(z-z^\circ)\,dF
=\int_{\theta^*}^1(z-\bar z)\,d(-M)
=\int_{\theta^*}^1 M(\theta)\,dz(\theta)\;\le\;0,
\]
the second equality by parts using $M(\theta^*)=M(1)=0$, the last because $M\le0$ and $z$ is nondecreasing. Claim 1 then delivers optimality. Existence and uniqueness of $\lambda^*$: along the plateau condition $\theta_\lambda=1-\sqrt{\Delta/(2\lambda)}$ is strictly increasing in $\lambda$, while $\theta_\lambda=\beta^{-1}(q^*+\Delta\beta_0/\lambda)$ is strictly decreasing, so the two curves cross at most once, and they cross on $(2\Delta,\infty)$ exactly when $\beta^{-1}(q^*+\beta_0/2)>\theta^*$---the conflict case. Binding obedience then determines $\bar\phi$ as minus the ratio of the two integrals in the theorem; $\bar\phi<1$ is the regime $\E[\beta]>q^*$, and $\bar\phi\ge0$ is exactly the self-obedience hypothesis $\int_0^{\theta_{\lambda^*}}(\beta-q^*)\,dF\le0$. When self-obedience fails at the plateau root, set $\bar\phi=0$ and let binding obedience alone pin the pool's end, $T:=\theta_{\lambda^*}=2(q^*-\beta_0)/\beta_1$, with the multiplier $\lambda^\circ=\Delta\beta_0/(q^*-\beta_0)$ placing the crossing of $s_\lambda$ exactly at $T$ (note $\beta(T)=2q^*-\beta_0>q^*$ on the regime, so $\lambda^\circ>0$). On $(\theta^*,T]$ the plateau bundle sits at $\bar z=1$ from above, where $V$ is differentiable with slope $\psi$, so $\gamma=\psi$ is forced; at $\theta=T$ the bundle reaches the kink $\bar z=1$, and on $(T,1]$ the supergradient is the interval $[\tilde\gamma,\psi]$. Choose $\gamma=\psi+\bar t\,s_{\lambda^\circ}/c$ there, the constant $\bar t=\Delta/\big(2\lambda^\circ(1-T)^2\big)$ placing the kink correction so that $\int_{\theta^*}^1\gamma\,dF=0$ in closed form---the plateau condition holding as an inequality at $(\lambda^\circ,T)$ is exactly $\lambda^\circ(1-T)^2\ge\Delta/2$, i.e.\ $\bar t\in(0,1]$, so the selection is admissible. Single crossing is forced: $\gamma$ vanishes at $\theta^*$, equals $\psi>0$ on $(\theta^*,T]$, and is linear on $(T,1]$; were it nonnegative throughout, its integral could not vanish, so it crosses zero exactly once and $M\le0$ follows as before. On $[0,\theta^*)$ the selection $\gamma=0$ remains admissible because $\tilde\gamma$ is linear there with both endpoint values nonnegative---$\tilde\gamma(\theta^*)=(\lambda^\circ/\Delta)(T-\theta^*)\ge0$ by the guard, and $\tilde\gamma(0)=\psi(0)+s_{\lambda^\circ}(0)/c=-1+2\beta_0/\beta_1>0$: the corner together with the guard forces $\beta_1<2\beta_0$, since failed self-obedience means the plateau root exceeds $T$, whence $(1-t_{\mathrm{root}})^2=\beta_1(t_{\mathrm{root}}-T/2)/(2\beta_0)>\beta_1T/(4\beta_0)$ while $(1-t_{\mathrm{root}})^2<(1-T)^2$, and $4(1-T)^2/T\le2$ for $T\ge\theta^*=\tfrac12$. (When $\theta_{\lambda^*}\le\theta^*$ instead, the same construction irons a decreasing step at the bottom and yields a single partial pool there; none of our applications encounter this case.)

\emph{Claim 4 (uniqueness and the binary case).} On $[0,\theta^*)$ the slopes of $V$ are strict ($\tilde\gamma>0$ below $z=c$, $\psi<0$ above), so $z=c$ a.e.\ is forced; on $(\theta^*,1)$, $M<0$ forces $dz=0$, so $z$ is a.e.\ constant, and the two pinning equations have a unique solution under the strictly increasing virtual value, forcing $z=\bar z$; the inner split is strict off $\theta_{\lambda^*}$. The optimum is therefore essentially unique. For part (iii): with two-point support $\{0,1\}$ and masses $(1-\mu_0,\mu_0)$, monotonicity of $z$ is exactly \eqref{eq:M} after the affine reparametrization of valuations, and the program \eqref{eq:program} coincides with the reduced program of Theorem~\ref{thm:binary}'s proof, so their optima coincide---at $\mu_0<q^*$, the supply distortion $q_1(\low)=1-(\bar\high-\bar\low)g_\high$; the certificate corroborates the agreement exactly. $\square$

\medskip
\noindent\textbf{Proof of Proposition~\ref{prop:converse}.} Let $(x,\phi)$ be feasible and deterministic, $\phi\in\{0,1\}$ a.e., and write $z=x+c\phi$. Four exhaustive cases. (a) If $\phi$ violates the inner split of Claim 2 on positive measure, the Lagrangian falls strictly short of its maximum, since $s_{\lambda^*}\neq0$ a.e. (b) If $z$ places increase-mass on the open plateau interior, where $M<0$ strictly, the verification chain of Claim 3 is strict. (c) If $z$ differs from $c$ on positive measure below $\theta^*$, the strict slopes of $V$ there ($\tilde\gamma>0$ below $z=c$, $\psi<0$ above) make the chain strictly slack. (d) Otherwise $z=c$ a.e.\ below $\theta^*$ and $z$ is constant on $(\theta^*,1]$, so by the inner split $\phi=1$ below $\theta_{\lambda^*}$ and $\phi$ is a constant $\bar\phi'\in\{0,1\}$ above. If $\bar\phi'=1$ then $\phi\equiv1$, violating obedience on the regime. If $\bar\phi'=0$, obedience is strictly slack by $\bar\phi>0$---the optimum's binding obedience minus the candidate's is $\bar\phi\int_{\theta_{\lambda^*}}^1(\beta-q^*)\,dF>0$---so the candidate's objective falls short of the Lagrangian bound by $\lambda^*$ times that slack, which is strictly positive. In every case the deterministic rule attains strictly less than the optimum. $\square$
\section{Proof of Theorem~\ref{thm:incidence}}\label{app:incidence}

\noindent\textbf{Proof of Theorem~\ref{thm:incidence}.} \emph{(i)} Under $g^*$ the source posterior is $1$ with probability $r=(\mu_0-\pi)/(1-\pi)$ (Bayes weight of the reveal message) and $\pi$ otherwise, so $q_i\in\{q_i^R,q_i^P\}$ with the stated values and means $q_i^0$. For $q_i^0\le q^*$: $D(q_i^0)=0$ and $W_i=-[r\,D(q_i^R)+(1-r)\,D(q_i^P)]$; $q_i^P\le q_i^0\le q^*$ gives $D(q_i^P)=0$, leaving $W_i=-r\low(1-q_i^R)\ind\{q_i^R>q^*\}$. For $q_i^0>q^*$: write $\tilde D(q)=\low(1-q)$ for the affine extension of $D$ above $q^*$; Bayes plausibility gives $D(q_i^0)=\tilde D(q_i^0)=r\tilde D(q_i^R)+(1-r)\tilde D(q_i^P)$, and since $q_i^R\ge q_i^0>q^*$ implies $D(q_i^R)=\tilde D(q_i^R)$,
\[
W_i=D(q_i^0)-r D(q_i^R)-(1-r)D(q_i^P)=(1-r)\big[\tilde D(q_i^P)-D(q_i^P)\big]=(1-r)\low(1-q_i^P)\ind\{q_i^P\le q^*\}.
\]
The crossing thresholds restate $q_i^R>q^*$ and $q_i^P\le q^*$. \emph{(ii)} On the harm side, for $\rho_i>\underline\rho_i$, $|W_i|=r\low(1-q_i^0-\rho_i(\bar\high-\bar\low)(1-\mu_0))$ is strictly decreasing in $\rho_i$, zero at $\rho_i\le\underline\rho_i$, with supremum $r\low(1-q^*)$ as $\rho_i\downarrow\underline\rho_i$ (not attained: at the threshold exactly, $q_i^R=q^*$ is discounted by the tie-breaking convention and $W_i=0$). On the gain side, for $\rho_i\ge\bar\rho_i$, $W_i=(1-r)\low(1-q_i^0+\rho_i(\bar\high-\bar\low)(\mu_0-\pi))$ is strictly increasing, and at $\rho_i=\bar\rho_i$ exactly, $q_i^P=q^*$ is discounted and $W_i=(1-r)\low(1-q^*)$ is attained. \emph{(iii)} Crossing buyers below $q^*$ have $W_i<0$ whenever $q_i^R<1$, and crossing buyers above have $W_i>0$; with mass on both, surplus moves from the first group to the second; the aggregate $\int W_i\,dF$ can vanish while both integrals are nonzero. $\square$

\section{Supplemental Material}\label{app:pop}

\noindent Proofs of Propositions 4--12 are in the Supplemental Material, together with the replication certificates; every numerical value in the paper regenerates from the one-command certificate there.

\begin{spacing}{1.0}

\end{spacing}

\end{document}